\documentclass[11pt]{article}
\usepackage[utf8]{inputenc}
\pdfoutput=1 
\usepackage{amsmath, amssymb, amsthm, dsfont}
\usepackage{enumitem}
\usepackage{graphicx}
\usepackage{subcaption}
\usepackage[dvipsnames]{xcolor}
\usepackage{hyperref, comment}
\usepackage[ruled,vlined,linesnumbered]{algorithm2e}
\usepackage[left=1.00in, right=1.00in, top=1.00in, bottom=1.00in]{geometry}
\usepackage{natbib}
\setcitestyle{authoryear,open={(},close={)}} 
\usepackage{tikz}
\usetikzlibrary{calc}

\usepackage[framemethod=tikz]{mdframed}
\mdfsetup{
	roundcorner=4pt,
	linewidth=.5,
	innerleftmargin=0pt,
	innerrightmargin=10pt,
	innertopmargin=10pt,
	innerbottommargin=10pt,
	skipabove=12pt,skipbelow=12pt
}

\allowdisplaybreaks

\usepackage{color}

\title{Strategizing against No-Regret Learners in First-Price Auctions}
\author{Aviad Rubinstein\thanks{Supported by NSF CCF-1954927, and a David and Lucile Packard Fellowship.}\\Stanford University\\\texttt{aviad@cs.stanford.edu} \and Junyao Zhao\thanks{Supported by NSF CCF-1954927.}\\ Stanford University\\\texttt{junyaoz@stanford.edu}}
\date{}

\def \A {\mathcal{A}}

\def \br {\textnormal{\sffamily br}}

\def \D {\mathcal{D}}

\def \eps {\varepsilon}
\def \ex {\textnormal{\sffamily ex}}

\def \fpa {\textnormal{\sffamily fpa}}
\def \G {\mathcal{G}}

\def \I {\mathcal{I}}
\def \J {\mathcal{J}}

\def \P {\textnormal{Pr}}
\def \poly {\textnormal{\sffamily poly}}
\def \R {\mathbb{R}}
\def \S {\mathcal{S}}

\def \Reg {\textnormal{\sffamily Reg}}
\def \T {\mathcal{T}}

\DeclareMathOperator*{\argmax}{arg\,max}
\DeclareMathOperator*{\E}{\mathbb{E}}

\newcommand{\floor}[1]{\left\lfloor #1 \right\rfloor}
\newcommand{\class}[1]{\langle #1 \rangle}
\newcommand{\spoly}[1]{\textnormal{$#1$-\sffamily poly}}
\DeclareMathAlphabet{\mathibb}{U}{bbold}{m}{n}
\newcommand{\1}{\mathibb{1}}

\newtheorem{theorem}{Theorem}[section]
\newtheorem{lemma}[theorem]{Lemma}
\newtheorem{corollary}[theorem]{Corollary}
\newtheorem{proposition}[theorem]{Proposition}
\newtheorem{claim}[theorem]{Claim}
\newtheorem{definition}[theorem]{Definition}

\newtheorem{example}[theorem]{Example}

\newtheorem{assumption}[theorem]{Assumption}

\begin{document}

\maketitle

\begin{abstract}
We study repeated first-price auctions and general repeated Bayesian games between two players, where one player, the learner, employs a no-regret learning algorithm, and the other player, the optimizer, knowing the learner's algorithm, strategizes to maximize its own utility. For a commonly used class of no-regret learning algorithms called mean-based algorithms, we show that (i) in standard (i.e., full-information) first-price auctions, the optimizer cannot get more than the Stackelberg utility -- a standard benchmark in the literature, but (ii) in Bayesian first-price auctions, there are instances where the optimizer can achieve much higher than the Stackelberg utility.

On the other hand, Mansour et al.~(2022) showed that a more sophisticated class of algorithms called no-polytope-swap-regret algorithms are sufficient to cap the optimizer's utility at the Stackelberg utility in any repeated Bayesian game (including Bayesian first-price auctions), and they pose the open question whether no-polytope-swap-regret algorithms are necessary to cap the optimizer's utility. For general Bayesian games, under a reasonable and necessary condition, we prove that no-polytope-swap-regret algorithms are indeed necessary to cap the optimizer's utility and thus answer their open question. For Bayesian first-price auctions, we give a simple improvement of the standard algorithm for minimizing the polytope swap regret by exploiting the structure of Bayesian first-price auctions.
\end{abstract}

\setcounter{page}{0}
\thispagestyle{empty}
\newpage

\section{Introduction}
Consider a repeated (Bayesian\footnote{Briefly, a Bayesian game between the optimizer and the learner~\citep{mansour2022strategizing} consists of two sets of actions for the two players respectively and a publicly known prior distribution of the context (the private type of the learner). At the start of the game, a context is sampled from the prior distribution and revealed only to the learner. Then, the optimizer and the learner choose their actions to play. The utilities they get depend on the actions they play and the context.}) game between two players, where one player, the {\em learner}, employs a no-regret learning algorithm, which is a widely adopted type of strategy, and the other player, the {\em optimizer}, knowing the learner's algorithm, plays optimally to maximize its own utility.~\citet{mansour2022strategizing} observed that regardless of which no-regret learning algorithm the learner uses, the optimizer is always able to obtain the Stackelberg utility%
\footnote{Informally, the Stackelberg utility is the maximum utility the optimizer (called {\em leader} in Stackelberg literature) can achieve in the following single-round game: The optimizer first chooses a (mixed) strategy to commit to, and then the learner (called {\em follower} in Stackelberg literature) plays its best response.} on average per round by playing some static fixed strategy. 
It is thus interesting to ask
\begin{enumerate}
    \item[(i)] whether the optimizer can beat this benchmark by exploiting the commonly used no-regret learning algorithms in a specific game of interest, and
    \item[(ii)] which no-regret learning algorithms are sufficient and necessary to cap the optimizer's utility at this benchmark in general.
\end{enumerate}

A specific game that is of particular interest in
online advertising is first-price auction. Indeed, first-price auctions have become a prominent format for online ads auctions recently~\citep{despotakis2021first}. This is a highly repeated auction setting, where the players (bidders) get feedback from the previous rounds and respond by changing their bids for the future rounds. No-regret learning algorithms are natural strategies for the bidders in such setting, which have been extensively studied in theory~\citep{balseiro2019contextual,han2020learning,han2020optimal,zhang2022leveraging,badanidiyuru2023learning} and also observed in practice~\citep{nekipelov2015econometrics}. Many widely adopted no-regret learning algorithms, e.g.,
multiplicative weights update (MWU), follow the perturbed leader (FTPL), EXP3, UCB, $\varepsilon$-Greedy, and their variants, belong to a natural class of algorithms called mean-based learning algorithms\footnote{Informally, a learning algorithm is mean-based if the algorithm, in almost every round, plays an action that has approximately best utility on average in the previous rounds (conditioned on the context if the game is Bayesian).}~\citep{braverman2018selling,deng2019strategizing,feng2021convergence}, which motivates the first main question we study in this paper:
\begin{quote}
In a repeated (Bayesian) single-item first-price auction between two bidders, if one bidder (the optimizer) knows that the other bidder (the learner) uses a mean-based learning algorithm, can the optimizer get higher than Stackelberg utility on average per round by exploiting the learner's algorithm?
\end{quote}
In our first main result (Theorem~\ref{thm:mean-based-restatement}), we give a complete answer to this question, which turns out to vary according to whether the learner's value of the item is certain. 
\begin{theorem}[Informal restatement of Theorem~\ref{thm:robust-standard-mean-based} and Theorem~\ref{thm:exploit-bayesian-mean-based}]\label{thm:mean-based-restatement}
In a standard full-information first-price auction repeated for $T$ rounds (i.e., the learner's value of the item is static fixed and publicly known), if the learner uses any mean-based no-regret learning algorithm, then the optimizer's optimal utility in $T$ rounds is no more than $V\cdot T+o(T)$, where $V$ is the Stackelberg utility of this first-price auction.

However, there exists a Bayesian first-price auction repeated for $T$ rounds (i.e., the learner's value of the item is private and sampled from a static fixed and publicly known prior distribution in each round), such that there exists a strategy that guarantees the optimizer a utility of $V'\cdot T$ for some $V'$ that is significantly higher than the Stackelberg utility of this Bayesian first-price auction regardless of which mean-based learning algorithm the learner uses.
\end{theorem}

We postpone the intuition of Theorem~\ref{thm:mean-based-restatement} to Section~\ref{section:non-bayesian} and Section~\ref{section:bayesian}, and proceed to a natural follow-up question: are there no-regret learning algorithms that cannot be exploited by the optimizer in Bayesian first-price auctions? The answer is yes, not just for Bayesian first-price auctions, but for all Bayesian games. Specifically,~\citet{mansour2022strategizing} showed that a powerful class of no-regret learning algorithms called no-polytope-swap-regret learning algorithms are sufficient to cap the optimizer's average utility per round at the Stackelberg utility for any repeated Bayesian game. We defer the exposition of polytope swap regret to the preliminary (Definition~\ref{def:S-poly-regret}), but note here that polytope swap regret is a generalization of swap regret in standard full-information games~\citep{blum2007external} to Bayesian games.~\citet{mansour2022strategizing} posed a natural open question, which is the second main question we strive to answer in this paper:
\begin{quote}
    Are no-polytope-swap-regret learning algorithms necessary to cap the optimizer's average utility per round at the Stackelberg utility for all repeated Bayesian games?
\end{quote}
Intrigued by their open question, we turn our attention to general Bayesian games. In our second main result (Theorem~\ref{thm:poly-swap-regret-restatement}), we show that no-polytope-swap-regret learning algorithms are indeed necessary to cap the the optimizer's average utility per round at the Stackelberg utility, under a reasonable assumption: there is no negligible context in the Bayesian game, i.e.,~the probability of each context in the Bayesian game is $\Omega(1)$, which in particular implies that there is no vacuous context that occurs with zero probability (this assumption is also necessary as we will discuss shortly).
\begin{theorem}[Informal restatement of Theorem~\ref{thm:exploitable-poly-swap-regret}]\label{thm:poly-swap-regret-restatement}
Consider only the Bayesian games where the learner's and the optimizer's utilities are within the range $[-1,1]$. If a learning algorithm $\A$ has $\Omega(T)$ polytope swap regret in some Bayesian game $\G$ repeated for $T$ rounds which has no negligible context, then there exists another Bayesian game $\G'$ repeated for $T$ rounds such that the optimizer can obtain a utility of $V(\G')\cdot T + \Omega(T)$ when playing against the learning algorithm $\A$, where $V(\G')$ is the Stackelberg utility of the Bayesian game $\G'$.
\end{theorem}

We complement Theorem~\ref{thm:poly-swap-regret-restatement} by showing that there exists a repeated Bayesian game with a negligible context in which some learning algorithm $\A$ has $\Omega(T)$ polytope swap regret, but there does not exist any repeated Bayesian game where the optimizer can get higher than Stackelberg utility on average per round when playing against $\A$ (Proposition~\ref{proposition:negligible_context}), and hence, the no-negligible-context assumption in Theorem~\ref{thm:poly-swap-regret-restatement} is necessary.

By being sufficient and necessary (under a reasonable assumption), no-polytope-swap-regret algorithms are arguably the right algorithms to be robust from the exploitation of the optimizer for general Bayesian games. Can we improve these algorithms (e.g., their regret bound and runtime) for specific Bayesian games such as Bayesian first-price auctions? In an additional result (Proposition~\ref{prop:better-alg}), we give a simple improvement by pruning the space of the strategies that needs to be considered by the no-polytope-swap-regret learning algorithms according to the specific Bayesian game, which leads to slightly improved regret bound and runtime when applied to Bayesian first-price auctions (Corollary~\ref{cor:first-price-better-alg}).

Finally, we briefly mention two technical aspects of the proof of Theorem~\ref{thm:poly-swap-regret-restatement}, which motivate two further directions. First, our construction of the repeated Bayesian game $\G'$ in Theorem~\ref{thm:poly-swap-regret-restatement} is not fully explicit. That is, we use min-max theorem to prove the existence of some parameters of $\G'$. Second, the construction of $\G'$ blows up the number of the optimizer's actions, i.e., $\G'$ has more actions for the optimizer than the original Bayesian game $\G$. Making things even more interesting, blowing up the number of the optimizer's actions turns out to be necessary. Specifically, we exhibit an instance where some learning algorithm $\A$ has $\Omega(T)$ polytope swap regret in some repeated Bayesian game $\G$, but in any repeated Bayesian game $\G'$ with the same number of actions for the optimizer as $\G$, the optimizer can only get Stackelberg utility on average per round when playing against $\A$ (Proposition~\ref{proposition:same_action}). These interesting aspects raise two further questions which we did not answer in this paper: Is there a fully explicit construction of $\G'$? What is the minimum number of the optimizer's actions we need in $\G'$ to prove Theorem~\ref{thm:poly-swap-regret-restatement}?
\subsection{Related work}
\paragraph{Strategizing against no-regret learners}~\citet{braverman2018selling} initiated the study of repeated Bayesian games between a strategic optimizer and a no-regret learner. They focus on the specific Bayesian game of a single revenue-maximizing seller auctioning a single item to a single buyer (where the auction format is the seller's action, and the bid is the buyer's action), and they show that the seller can extract the full value of the buyer if the buyer uses a mean-based learning algorithm (this result has recently been generalized to the setting with multiple buyers by~\citet{cai2023selling}).~\citet{deng2019strategizing} studies general standard (i.e., full-information) repeated games. They give an example of a standard repeated game where the optimizer can exploit mean-based learning algorithms (in contrast, our result implies that standard first-price auction is not such game), and they prove that no-swap-regret learning algorithms are sufficient to cap the optimizer's average utility per round at the Stackelberg utility for standard games.~\citet{mansour2022strategizing} complements the result of~\citet{deng2019strategizing} by showing that no-swap-regret learning algorithms are also necessary to cap the optimizer's utility for standard games (this result is a special case of our Theorem~\ref{thm:poly-swap-regret-restatement}, but interestingly, for this special case, their main construction is explicit, and it does not need to blow up the number of the optimizer's actions). They then started the study of general repeated Bayesian games and proposed no-polytope-swap-regret learning algorithms as a natural class of algorithms that are sufficient to prevent the optimizer's exploitation in general repeated Bayesian games, and they leave the open problem of whether these algorithms are necessary.

\paragraph{First-price auctions} There are significant amount of works studying how to maximize a bidder's utility in repeated first-price auctions using no-regret learning algorithms~\citep{balseiro2019contextual,han2020learning,han2020optimal,zhang2022leveraging,badanidiyuru2023learning}, but the interaction between a strategic optimizer and a learner is much less studied, except that~\citet{xu2018commitment} studied the Stackelberg utility of the single-round Bayesian first-price auction. In addition,~\citet{feng2021convergence} studies mean-based learning algorithms for repeated first-price auctions, but they focus on studying the equilibrium these algorithms converge to.

\section{Preliminaries}\label{section:preliminaries}
\subsection{Games and equilibria}
In general, we consider finite Bayesian bimatrix games which we refer to as \emph{Bayesian games}. In Bayesian games, there are two players who we refer to as \emph{optimizer} and \emph{learner}. The learner has a private \emph{type} sampled from a distribution, and the optimizer only knows the distribution of the learner's type but not the type itself. Formally, in a Bayesian game $\G(M,N,C,\D,u_O,u_L)$, there are $M$ actions $[M]$ for the optimizer and $N$ actions $[N]$ for the learner, and moreover, there are $C$ \emph{contexts} (i.e., types) $[C]$, and at the beginning of the game, a context $c\in [C]$ is sampled from a prior distribution $\D$, and we let $p_c$ denote the probability that context $c$ occurs. The prior distribution $\D$ is known to both players, but the sampled context $c$ is known only to the learner but not the optimizer. Therefore, the optimizer can only choose an action $i\in[M]$ that is independent of context $c$, but the learner may choose a specific action $f(c)\in[N]$ based on context $c$ (where $f\in[N]^{[C]}$ is a map from contexts to the learner's actions, and we call such map a \emph{pure strategy} for the learner), and then the optimizer receives utility $u_{O}(i,f(c),c)\in \R$, and the learner receives utility $u_{L}(i,f(c),c)\in \R$ (note that these are the players' utilities conditioned on the context, which therefore not only depend on the actions played by the players but also depend on the context).

In general, the optimizer can use a \emph{mixed strategy} $\alpha\in\Delta([M])$ which is a distribution over the optimizer's actions, and the learner can also use a mixed strategy $\beta\in\Delta([N]^{[C]})$ which is a distribution over the learner's pure strategies. Let $\alpha_i$ denote the probability of action $i$ in the optimizer's mixed strategy $\alpha$, and let $\beta_f$ denote the probability of pure strategy $f$ in the learner's mixed strategy $\beta$. Let $u_{O}(\alpha,f(c),c):=\sum_{i\in[N]}\alpha_i u_{O}(i,f(c),c)$ and $u_{L}(\alpha,f(c),c):=\sum_{i\in[N]}\alpha_i u_{L}(i,f(c),c)$. Then, the optimizer's expected utility $u_{O}(\alpha, \beta)$ and the learner's expected utility $u_{L}(\alpha, \beta)$ (where the expectation is over the randomness of their mixed strategies and the context) are
\begin{align*}
u_{O}(\alpha, \beta)=\sum_{f\in[N]^{[C]},c\in[C]}\beta_f p_c u_{O}(\alpha,f(c),c)\textrm{ and }
u_{L}(\alpha, \beta)=\sum_{f\in[N]^{[C]},c\in[C]}\beta_f p_c u_{L}(\alpha,f(c),c),
\end{align*}
respectively. Sometimes we are interested in the value $u_{O}(\alpha, \beta)$ or $u_{L}(\alpha, \beta)$ when $\alpha$ is  a trivial distribution with all the probability on a single action $i$, or $\beta$ is a trivial distribution on a single pure strategy $f$. In such cases, we use notations like $u_{L}(\alpha, f)$, $u_{L}(i, \beta)$, and $u_{L}(i, f)$ for convenience. For example,
\begin{align}\label{eq:trivial_mixed_strategy}
u_{L}(\alpha, f)=\sum_{c\in[C]} p_c u_{L}(\alpha,f(c),c).
\end{align}
Given the optimizer's mixed strategy $\alpha\in\Delta([M])$, we call the learner's pure strategy $f$ a \emph{best response} to the strategy $\alpha$ if $f\in \argmax_{f'\in[N]^{[C]}}u_{L}(\alpha,f')$, which by Eq.~\eqref{eq:trivial_mixed_strategy} is equivalent to that $f(c)\in \argmax_{f'(c)\in[N]}u_{L}(\alpha,f'(c),c)$ for all $c\in[C]$. In other words, conditioned on each context $c\in[C]$, $f(c)$ is a best response to the strategy $\alpha$.

Now we define the \emph{Stackelberg equilibrium} 
and the \emph{Stackelberg utility} of a Bayesian game. Intuitively, in a Stackelberg equilibrium, the optimizer chooses a mixed strategy that maximizes its own expected utility, assuming that the learner always plays the best response to its mixed strategy (and when there are multiple best responses, we assume the learner picks the one that maximizes the optimizer's utility). The Stackelberg utility is the optimizer's expected utility in the Stackelberg equilibrium.
\begin{definition}\label{def:bayesian_stackelberg}
The Stackelberg equilibrium of a Bayesian game $\G(M,N,C,\D,u_O,u_L)$ consists of a mixed strategy of the optimizer $\alpha\in\Delta([M])$ and a pure strategy of the learner $f\in [N]^{[C]}$, which maximize $u_{O}(\alpha, f)$ under the constraint that $f$ is a best response to the strategy $\alpha$. The Stackelberg utility of a Bayesian game $\G$, which we usually denote with $V(\G)$, is the value $u_{O}(\alpha, f)$, where $(\alpha, f)$ is a Stackelberg equilibrium.
\end{definition}

\subsubsection*{First-price auctions}
For a significant part of the paper, we will focus on a specific Bayesian game -- a single-item \emph{first-price auction} between the optimizer and the learner. In a first-price auction of a single item, the optimizer has a value $v_O\in\R_{\ge 0}$ for the item, and the learner has a value $v_L\in\R_{\ge 0}$ for the item. Each player submits a bid which represents the price they are willing to pay, and the bids that can be submitted by the players are the actions in this game. The learner's value of the item $v_L$ is the context in this game, which is sampled from a prior distribution $\D$ over $\I$ (in Assumption~\ref{assumption:discrete-bid}). We make the following assumption\footnote{Note that this assumption is very natural in practice and standard in the literature of no-regret algorithms for repeated auctions (e.g.,~\cite{braverman2018selling,feng2021convergence}).} such that this game has finite number of actions and contexts. 

\begin{assumption}\label{assumption:discrete-bid}
The players' values and bids for the item are in the discrete space $\I=\{i\cdot \eps\mid i\in \{0,\dots,N-1\}\}$ for some constant $\eps\in\mathbb{R}_{>0}$ (e.g., $\eps=0.01$) and parameter $N\in \mathbb{N}$. (We will also refer to the bid $i\cdot\eps$ as action $i$ and use the terms ``action'' and ``bid'' interchangeably.)
\end{assumption}
The player who bids the higher price wins the item and pays that higher price, and the other player gets nothing and pays nothing. For simplicity, we assume that the tie-breaking rule always favors the learner (our results in this paper can be easily adapted for any tie-breaking rule).
\begin{assumption}\label{assumption:tie-breaking}
The learner wins the item if there is a tie.
\end{assumption}
Conditioned on context $v_L$, the utility of the player who wins the item is their value of the item minus the price they pay, and the utility of the other player is zero.

As elaborated above, first-price auction is a well-defined Bayesian game, which we will refer to as \emph{Bayesian first-price auction}. We will also consider an important case of Bayesian first-price auction, the full-information setting where the context is fixed deterministically (namely, $\D$ is a trivial distribution on a single value $v_L$), and we will refer to this special case as \emph{standard first-price auction} (more generally, we refer to any Bayesian game with a fixed context as \emph{standard game}). Note that every definition we have for general Bayesian game applies to this special case.

\subsection{Repeated games and no-regret learners}
We are particularly interested in the online setting where the optimizer and the learner repeatedly play a Bayesian game $\G(M,N,C,\D,u_O,u_L)$ for $T$ rounds.
In this online setting, the parameters of $\G$ are public information to both players. In each round, a context is sampled from $\D$ independently (which is known to the learner but not the optimizer), then the optimizer and the learner choose their actions to play and receive their utilities as a result of the actions they played (the actions they played and the learner's context at this round are made public after the players receive their utilities, but not before that).
\begin{assumption}\label{assumption:finite_runtime}
We assume both the optimizer's and the learner's algorithms have finite runtime.
\end{assumption}
We do not have additional restriction on the optimizer -- it is allowed to be randomized and adaptive, and it knows the learner's algorithm (but does not know the randomness used by the learner's algorithm, i.e., if the learner flips random coins to decide which action to play at a round, then the optimizer does not know the outcome of those random coin flips before playing its action). On the other hand, we assume the learner, as the name suggests, uses a ``no-regret'' learning algorithm. There are different notions of ``regret'' in the literature with different strength. We start by defining the most basic notion called \emph{external regret}.
\subsubsection*{External regret and mean-based learners}
Intuitively, external regret quantifies the extra utility the learner would have received if it had played the best single pure strategy in the hindsight.
\begin{definition}[external regret]\label{def:external-regret}
In a Bayesian game $\G(M,N,C,\D,u_O,u_L)$ repeated for $T$ rounds, suppose that at each round $t\in[T]$, the optimizer plays action\footnote{The optimizer can be adaptive and randomized, and $i_t$ is the action the optimizer eventually plays at round $t$.} 
$i_{t}\in[M]$, and the learner uses mixed strategy $\beta_{t}\in\Delta([N]^{[C]})$. Then, the learner's external regret is $\Reg_{\ex}:=\max_{f'\in[N]^{[C]}}\sum_{t\in[T]}u_L(i_t,f')-\sum_{t\in[T]} u_L(i_t,\beta_{t})$. Moreover, we say that the learner is no-external-regret if $\Reg_{\ex}=o(T)$.
\end{definition}

A special class of no-external-regret learners which we are interested in is the class of \emph{mean-based} no-external-regret learners introduced by~\citet{braverman2018selling}. Intuitively, at each round of the repeated game, a mean-based learner with high probability plays an action that approximately performs the best based on the current history. It was shown that many no-external-regret algorithms, including commonly used variants of EXP3, the Multiplicative Weights algorithm and the Follow-the-Perturbed-Leader algorithm are mean-based~\citep{braverman2018selling}.
\begin{definition}[mean-based learner]\label{def:mean-based}
In a Bayesian game $\G(M,N,C,\D,u_O,u_L)$ repeated for $T$ rounds, suppose that at each round $t\in[T]$, the optimizer plays action $i_{t}\in[M]$. For each $t\in[T]$, $c\in[C]$, and $j\in[N]$, let $\sigma_{j,t}(c):=\sum_{s=1}^t u_L(i_{s}, j, c)$ be the cumulative utility of action $j$ for context $c$ for the first $t$ rounds, and let $c_t\in[C]$ be the context at round $t$. Then, the learner is $\gamma$-mean-based if for all $t\in[T]$, the probability of playing action $j\in[N]$ at round $t$ is at most $\gamma$ whenever there exists $j'\in[N]$ such that $\sigma_{j,t}(c_t)<\sigma_{j',t}(c_t) - \gamma T$. Moreover, we say that the learner is mean-based if it is $\gamma$-mean-based for some $\gamma$ such that $\gamma =o(1)$.
\end{definition}
Following Definition~\ref{def:mean-based}, we define a few more terms which will be useful in the analysis. We say that \emph{action $j\in[N]$ is $\gamma$-mean-based at round $t$} if there does not exist $j'\in[N]$ such that $\sigma_{j,t}(c_t)<\sigma_{j',t}(c_t) - \gamma T$, and moreover, we say that \emph{round $t$ is $\gamma$-mean-based} if the learner eventually plays a $\gamma$-mean-based action at round $t$.

Next, we introduce a stronger notion called \emph{polytope swap regret}~\citep{mansour2022strategizing}.
\subsubsection*{Polytope swap regret}
For any mixed strategies $\beta_t,\rho_t\in\Delta([N]^{[C]})$ of the learner, let $\beta_{t,f},\rho_{t,f}$ denote the probabilities of pure strategy $f$ in mixed strategies $\beta_t,\rho_t$ respectively. We say $\beta_t$ is \emph{equivalent to} $\rho_t$ if for all $c\in[C]$ and $j\in[N]$, $\sum_{f\in[N]^{[C]}\textrm{ s.t.~}f(c)=j}\beta_{t,f}=\sum_{f\in[N]^{[C]}\textrm{ s.t.~}f(c)=j}\rho_{t,f}$, i.e., the probability that the learner plays action $j$ conditioned on context $c$ is the same regardless of whether the learner uses mixed strategy $\beta_t$ or $\rho_t$. Moreover, we can define an equivalent class $\class{\beta_t}:=\{\rho_t\in\Delta([N]^{[C]})\mid\textrm{$\rho_t$ is equivalent to $\beta_t$}\}$. Now we define \emph{$\S$-polytope swap regret} which generalizes the polytope regret.
\begin{definition}[polytope swap regret]\label{def:S-poly-regret}
In a Bayesian game $\G(M,N,C,\D,u_O,u_L)$ repeated for $T$ rounds, suppose that at each round $t\in[T]$, the optimizer plays action $i_{t}\in[M]$, and the learner uses mixed strategy $\beta_{t}\in\Delta([N]^{[C]})$. For any $\S\subseteq[N]^{[C]}$, the learner's $\S$-polytope swap regret is
$$\Reg_{\spoly{\S}}:=\min_{\forall t\in[T],\,\rho_t\in\class{\beta_t}}\max_{\pi:[N]^{[C]}\to\S}\sum_{t\in[T]}\sum_{f\in[N]^{[C]}}\rho_{t,f}\cdot u_L(i_t,\pi(f))-u_L(i_t,\beta_t).$$
The polytope swap regret of~\cite{mansour2022strategizing}, which we denote by $\Reg_{\poly}$, is the $\S$-polytope swap regret for $\S=[N]^{[C]}$.
Moreover, we say that the learner is no-$\S$-polytope-swap-regret if $\Reg_{\spoly{\S}}=o(T)$, and we say that it is no-polytope-swap-regret if $\Reg_{\poly}=o(T)$.
\end{definition}
Intuitively, since a mixed strategy of the learner is a distribution over pure strategies, $\S$-polytope swap regret quantifies the extra utility the learner would have received if it had replaced every pure strategy $f\in[N]^{[C]}$ with another pure strategy $\pi(f)\in\S$ simultaneously in all the mixed strategies $\beta_t$'s used by the learner in $T$ rounds. However, there are many mixed strategies equivalent to $\beta_t$, and crucially, Definition~\ref{def:S-poly-regret} permissively chooses the mixed strategies $\rho_t$'s such that the learner's regret is minimum. We also note that $\Reg_{\poly}\ge\Reg_{\ex}$ always holds because the mapping $\pi$ could map everything to the same pure strategy $f'\in[N]^{[C]}$. Furthermore, a nice property of no-polytope-swap-regret learners is that they can cap the optimizer's utility in any repeated Bayesian games:
\begin{lemma}[{\citet[Theorem 6]{mansour2022strategizing}}]\label{lem:mansour-theorem}
In any Bayesian game $\G$ repeated for $T$ rounds, if the learner is no-polytope-swap-regret, then the optimizer's optimal expected utility in $T$ rounds is no more than $V\cdot T+o(T)$, where $V$ is the Stackelberg utility of $\G$.
\end{lemma}
It is worth noting that the optimizer can always achieve expected utility as much as the Stackelberg utility times the number of rounds against any no-external-regret learner in any Bayesian game~\citep[Lemma 15]{mansour2022strategizing}, which makes the Stackelberg utility a natural benchmark to grade the optimizer's performance. Thus, Lemma~\ref{lem:mansour-theorem} essentially says that by employing a no-polytope-swap-regret algorithm, the learner can cap the optimizer's utility at this benchmark.

\section{Robust mean-based learners in standard first-price auction}\label{section:non-bayesian}
In this section, we show that mean-based no-external-regret learner can cap the optimizer's utility at the Stackelberg utility in any repeated standard (i.e., full-information) first-price auction, where the learner's value is fixed and publicly known.
\begin{theorem}\label{thm:robust-standard-mean-based}
In any standard first-price auction repeated for $T$ rounds, if the learner is no-external-regret and $\gamma$-mean-based such that $\gamma N^2=o(1)$, then the optimizer's optimal expected utility in $T$ rounds is no more than $V\cdot T+o(T)$, where $V$ is the Stackelberg utility of the standard first-price auction.
\end{theorem}
We first give a high-level proof sketch and then prove Theorem~\ref{thm:robust-standard-mean-based}. 
\begin{proof}[High-level proof sketch]\renewcommand{\qedsymbol}{}
Given an arbitrary no-external-regret mean-based learner, we want to upper bound the optimizer's expected utility in $T$ rounds using the Stackelberg utility. In general, the optimizer can be randomized and adaptive. We first show that w.l.o.g.~the optimizer is {\em deterministic} by a standard probabilistic argument. Then, we construct a sequence of bids which is {\em oblivious} (aka independent of learner's bids in previous rounds), and yet  achieves expected utility almost as good as the optimizer's deterministic algorithm. Intuitively, we are able to do this because the behaviour of a mean-based learner is reasonably predictable, and thus adaptivity does not help the optimizer much. Moreover, we show that the oblivious sequence of bids we constructed has a nice property -- the learner's best response to the optimizer's mixed strategy $\hat{\alpha}$ which corresponds to the empirical distribution of the oblivious sequence of bids, is essentially bidding $0$ (this is where we will use the assumption that the learner is no-external-regret). Note that (barring the case where the optimizer bids higher than its value, which will be handled in the proof) the optimizer's total expected utility using the oblivious sequence of bids against the no-external-regret learner, is trivially upper bounded by its total utility using the oblivious sequence of bids against a learner who always bids $0$. The latter is exactly the expected utility of the optimizer's mixed strategy $\hat{\alpha}$ against the learner's best response to $\hat{\alpha}$ times the number of rounds. Finally, note that the expected utility of the optimizer's mixed strategy $\hat{\alpha}$ against the learner's best response is at most the Stackelberg utility by definition of the Stackelberg utility.
\end{proof}
\begin{proof}[Proof of Theorem~\ref{thm:robust-standard-mean-based}]
Suppose the optimizer's value and the learner's value of the item are $v_O, v_L\in\I$ respectively ($\I$ is defined in Assumption~\ref{assumption:discrete-bid}, and the learner's value $v_L$ is fixed since the auction is standard), and suppose the learner is $\gamma$-mean-based for some $\gamma$ such that $\gamma N^2=o(1)$. In general, the optimizer can be randomized and adaptive, and we first reduce to the case where the optimizer is deterministic and oblivious.
\subsubsection*{Reducing to deterministic optimizer}
Suppose the optimizer runs a randomized algorithm $\A_O^{(\tilde{r}_O)}$ which decides which action to play (i.e., which bid to submit) for each round given access to a uniformly random binary string $\tilde{r}_O\sim\textrm{Uniform}(\{0,1\}^{R})$, and similarly, the learner runs a $\gamma$-mean-based no-external-regret randomized algorithm $\A_L^{(\tilde{r}_L)}$ given a uniformly random binary string $\tilde{r}_L\sim\textrm{Uniform}(\{0,1\}^{R})$, where we let $R$ be an upper bound of the total runtime of the optimizer's and the learner's algorithms in $2T$ rounds (such upper bound exists by Assumption~\ref{assumption:finite_runtime}). For any fixed binary strings $r_O$ and $r_L$, the algorithms $\A_O^{(r_O)}$ and $\A_L^{(r_L)}$ become deterministic, and hence, the utility of the optimizer in all rounds $[T]$ is determined, which we denote by $U_O(\A_O^{(r_O)},\A_L^{(r_L)},[T])$. Thus, the expected utility of the optimizer in $T$ rounds can be represented as $\E_{\tilde{r}_O}\E_{\tilde{r}_L}[U_O(\A_O^{(\tilde{r}_O)},\A_L^{(\tilde{r}_L)},[T])]$, which is at most $\max\limits_{r_O\in\{0,1\}^{|\tilde{r}_O|}}\E_{r_L}[U_O(\A_O^{(r_O)},\A_L^{(\tilde{r}_L)},[T])]$.
Therefore, there exists a determinsitically fixed binary string $r_O^*\in\{0,1\}^{|\tilde{r}_O|}$ such that the deterministic algorithm $\A_O^{(r_O^*)}$ achieves no less expected utility in $T$ rounds than the randomized algorithm $\A_O^{(\tilde{r}_O)}$ for the optimizer against the learner's algorithm $\A_L^{(\tilde{r}_L)}$. Henceforth, we assume that the optimizer runs some deterministic algorithm $\A_O$ against the learner's algorithm $\A_L^{(\tilde{r}_L)}$ w.l.o.g.

\subsubsection*{Reducing to oblivious optimizer}
Let $\T(\tilde{r}_L)=\{t\in[T]\textrm{ s.t.~round $t$ is $\gamma$-mean-based}\}$ (note that $\T(\tilde{r}_L)$ is a random set which depends on the randomness $\tilde{r}_L$ used by the learner's algorithm $\A_L^{(\tilde{r}_L)}$). By definition of $\gamma$-mean-based learner, if an action $j\in\{0,\dots,N-1\}$ is not $\gamma$-mean-based, then the probability that $j$ is played is at most $\gamma$. The number of actions that are not $\gamma$-mean-based is trivially upper bounded by $N$ (in Assumption~\ref{assumption:discrete-bid}), and hence, by a union bound, for each $t\in[T]$, the probability that round $t$ is not $\gamma$-mean-based is at most $\gamma N$. It follows that $\E_{\tilde{r}_L}[|[T]\setminus\T(\tilde{r}_L)|]\le \gamma N T$. Now we decompose the optimizer's expected utility in $T$ rounds as follows
\begin{align*}
    \E_{\tilde{r}_L}[U_O(\A_O,\A_L^{(\tilde{r}_L)},[T])]=&\E_{\tilde{r}_L}[\underbrace{U_O(\A_O,\A_L^{(\tilde{r}_L)},\T(\tilde{r}_L))}_{\textrm{optimizer's utility in rounds $\T(\tilde{r}_L)$}}+\underbrace{U_O(\A_O,\A_L^{(\tilde{r}_L)},[T]\setminus\T(\tilde{r}_L))}_{\textrm{optimizer's utility in rounds $[T]\setminus\T(\tilde{r}_L)$}}]\\
    \le&\E_{\tilde{r}_L}[U_O(\A_O,\A_L^{(\tilde{r}_L)},\T(\tilde{r}_L))+v_O \cdot |[T]\setminus\T(\tilde{r}_L)|]\\
    &\text{(Because optimizer's utility per round is at most $v_O$)}\\
    = &\E_{\tilde{r}_L}[U_O(\A_O,\A_L^{(\tilde{r}_L)},\T(\tilde{r}_L))] + v_O \cdot\E_{\tilde{r}_L}[|[T]\setminus\T(\tilde{r}_L)|]\\
    \le &\E_{\tilde{r}_L}[U_O(\A_O,\A_L^{(\tilde{r}_L)},\T(\tilde{r}_L))] + v_O\cdot \gamma N T\\
    = &\E_{\tilde{r}_L}[U_O(\A_O,\A_L^{(\tilde{r}_L)},\T(\tilde{r}_L))] + o(T)\\
    &\text{(Because $v_O\le\eps N$ by Assumption~\ref{assumption:discrete-bid} and $\gamma N^2=o(1)$)}\\
    \le&\max\limits_{r_L\in\{0,1\}^{R}}U_O(\A_O,\A_L^{(r_L)},\T(r_L)) + o(T).
\end{align*}
Let $r_L^*:=\argmax_{r_L\in\{0,1\}^{R}}U_O(\A_O,\A_L^{(r_L)},\T(r_L))$. Given $r_L^*$, we now construct an oblivious algorithm for the optimizer that achieves expected utility at least $U_O(\A_O,\A_L^{(r_L^*)},\T(r_L^*))-o(T)$ (and hence at least $\E_{\tilde{r}_L}[U_O(\A_O,\A_L^{(\tilde{r}_L)},[T])]-o(T)$) in $T+o(T)$ rounds against the learner's randomized algorithm $\A_L^{(\tilde{r}_L)}$ (and we will then upper bound the expected utility achieved by this oblivious algorithm in $T+o(T)$ rounds in order to prove the theorem).

Notice that both algorithms $\A_O$ and $\A_L^{(r_L^*)}$ are deterministic. Thus, together they result in a deterministic sequence of bids $b_O^1,\dots,b_O^T\in\I$ from the optimizer and another deterministic sequence of bids $b_L^1,\dots,b_L^T\in\I$ from the learner. We denote by $X_t(b[1:T'])\subseteq\{0,\dots,N-1\}$ the set of $\gamma$-mean-based actions of the learner at round $t\in [T']$ if the auction is repeated for $T'$ rounds with the optimizer bidding $b^1,\dots,b^{T'}$. The oblivious algorithm we construct for the optimizer is simply bidding $\hat{b}_{O}^1,\dots,\hat{b}_{O}^{(1+\hat{\gamma})T}$ in $(1+\hat{\gamma})T$ rounds, where $\hat{b}_{O}^t=0$ for all $t\in[\hat{\gamma}T]$, and in the remaining rounds, we let $\hat{b}_{O}^{t+\hat{\gamma}T}=b_O^{t}$ if $b_O^{t}$ is higher than the minimum bid that is $\gamma$-mean-based at round $t$ but not higher than the optimizer's value $v_O$, formally,
$$\hat{b}_{O}^{t+\hat{\gamma}T}=\1(\min(X_t(b_O[1:T]))\cdot\eps<b_O^{t}\le v_O)\cdot b_O^{t}  \;\;\; \forall t\in [T],$$ where $\varepsilon$ is the constant in Assumption~\ref{assumption:discrete-bid}, and $\hat{\gamma}:=\frac{3\gamma}{\eps-\gamma}=o(1)$. We want to show that when bidding this sequence against the learner's randomized algorithm $\A_L^{(\tilde{r}_L)}$ in $(1+\hat{\gamma})T$ rounds, the optimizer's expected utility in $(1+\hat{\gamma})T$ rounds is at least $U_O(\A_O,\A_L^{(r_L^*)},\T(r_L^*))-o(T)$. The following claim is the key property we will use to show this. Henceforth, we will use $X_t$ and $\hat{X}_{t+\hat{\gamma}T}$ to refer to $X_t(b_O[1:T])$ and $X_{t+\hat{\gamma}T}(\hat{b}_O[1:(1+\hat{\gamma})T])$ for simplicity.

\begin{claim}\label{claim:lower-mean-based-bid}
Suppose $\gamma=o(1)$ and $\hat{\gamma}>\frac{2\gamma}{\eps-\gamma}$ (hence $\eps\hat{\gamma}>\gamma(2+\hat{\gamma})$). Consider any two sequences of the optimizer's bids $b_O^1,\dots,b_O^T\in\I$ and $\hat{b}_O^1,\dots,\hat{b}_O^{(1+\hat{\gamma})T}\in\I$, which satisfy that $\hat{b}_O^t=0$ for each $t\in[\hat{\gamma} T]$, and
$\hat{b}_O^t$ is either $b_O^{t-\hat{\gamma}T}$ or $0$ for each $t\in\{\hat{\gamma} T+1,\dots,(1+\hat{\gamma})T\}$.
Then, it holds for all $t\in[T]$ that $\max(\hat{X}_{t+\hat{\gamma}T})\le\min (X_t)$ (in other words, the maximum bid that is $\gamma$-mean-based at round $t+\hat{\gamma}T$ if the auction is repeated for $(1+\hat{\gamma})T$ rounds with the optimizer bidding $\hat{b}_O^1,\dots,\hat{b}_O^{(1+\hat{\gamma})T}$, is no higher than the minimum bid that is $\gamma$-mean-based at round $t$ if the auction is repeated for $T$ rounds with the optimizer bidding $b_O^1,\dots,b_O^T$).
\end{claim}
\begin{proof}[Proof of Claim~\ref{claim:lower-mean-based-bid}]
Let $r_{j,s}$ be the utility of action $j\in\{0,\dots,N-1\}$ for the learner at round $s\in[T]$ if the optimizer bids $b_O^{s}$ at round $s$, and let $\hat{r}_{j,s}$ be the utility of action $j$ for the learner at round $s\in[(1+\hat{\gamma})T]$ if the optimizer bids $\hat{b}_O^{s}$ at round $s$.
Let $\sigma_{j,t}=\sum_{s=1}^t r_{j,s}$ denote the cumulative utility of action $j$ for the first $t$ rounds given the optimizer's bids $b_O^1,\dots,b_O^{t}$, and let $\hat{\sigma}_{j,t}=\sum_{s=1}^t \hat{r}_{j,s}$ denote the cumulative utility of action $j$ for the first $t$ rounds given the optimizer's bids $\hat{b}_O^1,\dots,\hat{b}_O^{s}$.
Consider any $j\in\{0,\dots,N-1\}$ such that $j>\min(X_t)$. Our goal is to show that $j\notin\hat{X}_{t+\hat{\gamma}T}$ for any $j>\min(X_t)$ (which implies $\max(\hat{X}_{t+\hat{\gamma}T})\le \min(X_t)$). We prove this in two steps: Let $j_t=\min(X_t)$ and consider any $j>j_t$. First, we prove that at each round, action $j$ is not better than $j_t$. Then, we prove that in the first $\hat{\gamma}T$ rounds when the optimizer bids $0$, $j$ is strictly worse than $j_t$, and the gap is large enough to imply that $j\notin \hat{X}_{t+\hat{\gamma}T}$.

\subsubsection*{Step 1: at each round, $j$ is not better than $j_t$}
In the first step, we compare the contribution of each round $s\in[t]$ to $\sigma_{j_t,t}-\sigma_{j,t}$ (which is $r_{j_t,s}-r_{j,s}$) with the contribution of round $s+\hat{\gamma}T$ to $\hat{\sigma}_{j_t,t+\gamma T}-\hat{\sigma}_{j,t+\gamma T}$ (which is $\hat{r}_{j_t,s+\hat{\gamma}T}-\hat{r}_{j,s+\hat{\gamma}T}$), and we want to show that w.l.o.g.~$r_{j_t,s}-r_{j,s}\le \hat{r}_{j_t,s+\hat{\gamma}T}-\hat{r}_{j,s+\hat{\gamma}T}$ for all $s\in[t]$. We show this by case analysis -- at each round $s\in[t]$, there are three possible cases: (a) $j_t\cdot\eps\ge b_O^s$, (b) $j\cdot\eps< b_O^s$, and (c) $j_t\cdot\eps< b_O^s\le j\cdot\eps$.

\paragraph{Case (a): $j_t\cdot\eps\ge b_O^s$}
In case (a), the two bids $j_t\cdot\eps$ and $j\cdot\eps$ (corresponding to the two actions $j_t$ and $j$) both win the item against the optimizer's bid $b_O^s$ (by Assumption~\ref{assumption:tie-breaking}). Moreover, since $\hat{b}_O^{s+\hat{\gamma}T}$ is either $b_O^s$ or $0$, both bids $j_t\cdot\eps$ and $j\cdot\eps$ should also win the item against the optimizer's bid $\hat{b}_O^{s+\hat{\gamma}T}$. Thus, in this case, we have that $r_{j_t,s}-r_{j,s}=\hat{r}_{j_t,s+\hat{\gamma}T}-\hat{r}_{j,s+\hat{\gamma}T}$. 

\paragraph{Case (b):  $j\cdot\eps< b_O^s$}
In case (b), both bids $j_t\cdot\eps$ and $j\cdot\eps$ lose the item against the optimizer's bid $b_O^s$, and hence $r_{j_t,s}-r_{j,s}=0$. If $\hat{b}_O^{s+\hat{\gamma}T}=b_O^s$, then obviously $\hat{r}_{j_t,s+\hat{\gamma}T}-\hat{r}_{j,s+\hat{\gamma}T}=r_{j_t,s}-r_{j,s}$. If $\hat{b}_O^{s+\hat{\gamma}T}=0$, both bids $j_t\cdot\eps$ and $j\cdot\eps$ win the item against the optimizer's bid $\hat{b}_O^{s+\hat{\gamma}T}$ (by Assumption~\ref{assumption:tie-breaking}), and hence $\hat{r}_{j_t,s+\hat{\gamma}T}-\hat{r}_{j,s+\hat{\gamma}T}=(v_L-j_t\cdot\eps)-(v_L-j\cdot\eps)=(j-j_t)\eps>0=r_{j_t,s}-r_{j,s}$.

\paragraph{Case (c): $j_t\cdot\eps< b_O^s\le j\cdot\eps$}
Case (c) is slightly more complicated due to the possibility that both bids $j_t\cdot\eps$ and $j\cdot\eps$ can be strictly higher than the learner's value $v_L$. However, if $j\cdot\eps>v_L$ (and hence $j\cdot\eps-v_L\ge\eps$ by Assumption~\ref{assumption:discrete-bid}), we can directly show that $j\notin\hat{X}_{t+\hat{\gamma}T}$ which is our final goal. Specifically, since $j\cdot\eps>v_L$, bid $j\cdot\eps$ either wins the item and gets utility $v_L-j\cdot\eps\le-\eps$ or loses the item and gets utility $0$, and hence $\hat{r}_{j,s'}\le 0$ for all $s'\in[t+\hat{\gamma}T]$. Moreover, because the optimizer's bid $\hat{b}_O^{s'}$ is $0$ for all $s'\in[\hat{\gamma} T]$, the learner's bid $j\cdot\eps$ wins the item and gets utility $\hat{r}_{j,s'}=v_L-j\cdot\eps\le-\eps$ in every round $s'\in[\hat{\gamma} T]$. Thus, $\hat{\sigma}_{j,t+\gamma T}=\sum_{s'=1}^{\hat{\gamma} T}\hat{r}_{j,s'}+\sum_{s'=\hat{\gamma} T+1}^{t+\gamma T}\hat{r}_{j,s'}\le-\eps \hat{\gamma} T$, which is strictly less than $-\gamma(2+\hat{\gamma})T$ due to the condition on $\hat{\gamma}$ in the claim. On the other hand, we know that $\hat{\sigma}_{0,t+\gamma T}\ge 0$ because bidding $0$ always results in non-negative utility. Therefore, we have that $\hat{\sigma}_{j,t+\gamma T}<\hat{\sigma}_{0,t+\gamma T}-\gamma(2+\hat{\gamma})T$ which implies $j\notin\hat{X}_{t+\hat{\gamma}T}$.

Now we can assume w.l.o.g.~that $j\cdot\eps\le v_L$ (and $j_t\cdot\eps\le v_L$ as $j_t<j$). In case (c), bid $j\cdot\eps$ wins the item against the optimizer's bid $b_O^s$ and gets utility $v_L-j\cdot\eps\ge0$, and bid $j_t\cdot\eps$ loses the item against $b_O^s$ and gets utility $r_{j_t,s}=0$. It follows that $r_{j_t,s}-r_{j,s}\le 0$. If $\hat{b}_O^{s+\hat{\gamma}T}=b_O^s$, then obviously $\hat{r}_{j_t,s+\hat{\gamma}T}-\hat{r}_{j,s+\hat{\gamma}T}=r_{j_t,s}-r_{j,s}$. If $\hat{b}_O^{s+\hat{\gamma}T}=0$, then both bids $j_t\cdot\eps$ and $j\cdot\eps$ win the item against the optimizer's bid $\hat{b}_O^{s+\hat{\gamma}T}$, and thus $\hat{r}_{j_t,s+\hat{\gamma}T}-\hat{r}_{j,s+\hat{\gamma}T}=(v_L-j_t\cdot\eps)-(v_L-j\cdot\eps)>0\ge r_{j_t,s}-r_{j,s}$.

\subsubsection*{Step 2: $j_t$ is strictly better than $j$ in the first $\hat{\gamma}T$ rounds when the optimizer bids $0$}
In the second step, we prove that $\hat{\sigma}_{j_t,t+\gamma T}-\hat{\sigma}_{j,t+\gamma T}>\gamma(1+\hat{\gamma})T$ which implies that $j\notin\hat{X}_{t+\hat{\gamma}T}$. We notice that $\hat{\sigma}_{j_t,t+\gamma T}-\hat{\sigma}_{j,t+\gamma T}=\left(\sum_{s'=1}^{\hat{\gamma} T}\hat{r}_{j_t,s'}-\hat{r}_{j,s'}\right)+\left(\sum_{s'=\hat{\gamma} T+1}^{t+\gamma T}\hat{r}_{j_t,s'}-\hat{r}_{j,s'}\right)$, and we lower bound $\sum_{s'=1}^{\hat{\gamma} T}\hat{r}_{j_t,s'}-\hat{r}_{j,s'}$ and $\sum_{s'=\hat{\gamma} T+1}^{t+\gamma T}\hat{r}_{j_t,s'}-\hat{r}_{j,s'}$ separately.

Since $j_t\in X_t$, we have that $\sigma_{j_t,t}\ge\sigma_{j,t}-\gamma T$, namely $\sum_{s=1}^t r_{j_t,s}-r_{j,s}\ge-\gamma T$. In step (i), we proved that $r_{j_t,s}-r_{j,s}\le \hat{r}_{j_t,s+\hat{\gamma}T}-\hat{r}_{j,s+\hat{\gamma}T}$ for all $s\in[t]$. Hence, we have that $\sum_{s=\hat{\gamma}T+1}^{t+\hat{\gamma}T}\hat{r}_{j_t,s}-\hat{r}_{j,s}=\sum_{s=1}^t\hat{r}_{j_t,s+\hat{\gamma}T}-\hat{r}_{j,s+\hat{\gamma}T}\ge\sum_{s=1}^t r_{j_t,s}-r_{j,s}\ge-\gamma T$.

On the other hand, when faced with the optimizer's bid $\hat{b}_O^{s}=0$ at each round $s\in[\hat{\gamma} T]$, bid $j\cdot\eps$ wins the item and gets utility $\hat{r}_{j,s}=v_L-j\cdot\eps$, and bid $j_t\cdot\eps$ also wins the item and gets utility $\hat{r}_{j_t,s}=v_L-j_t\cdot\eps$. Thus, $\sum_{s=1}^{\hat{\gamma} T}\hat{r}_{j_t,s}-\hat{r}_{j,s}=\hat{\gamma} T(j-j_t)\eps\ge \hat{\gamma} T\eps$, which is strictly greater than $\gamma(2+\hat{\gamma})T$ due to the condition on $\hat{\gamma}$ in the claim.

\subsubsection*{Putting Steps 1 and 2 together}
To sum up, we have that $\left(\sum_{s'=1}^{\hat{\gamma} T}\hat{r}_{j_t,s'}-\hat{r}_{j,s'}\right)+\left(\sum_{s'=\hat{\gamma} T+1}^{t+\gamma T}\hat{r}_{j_t,s'}-\hat{r}_{j,s'}\right)>-\gamma T + \gamma(2+\hat{\gamma})T=\gamma(1+\hat{\gamma})T$, which finishes the proof of the claim.
\end{proof}

Observe that the two sequences of bids $b_{O}^1,\dots,b_{O,}^{T}$ and $\hat{b}_{O}^1,\dots,\hat{b}_{O}^{(1+\hat{\gamma})T}$ we chose satisfy the condition in Claim~\ref{claim:lower-mean-based-bid}, and the $\hat{\gamma}$ we chose also satisfies the condition in the claim. Hence, we can apply Claim~\ref{claim:lower-mean-based-bid} and get that for all $t\in[T]$, $\max(\hat{X}_{t+\hat{\gamma}T})\le\min (X_{t})$, and
in particular, this holds for each $t\in\T(r_L^*)$.
Recall that by definition of $\T(r_L^*)$ and $b_L^1,\dots,b_L^T$, at each round $t\in\T(r_L^*)$, bid $b_{L}^t$ is $\gamma$-mean-based, i.e., $b_{L}^t=j_t\cdot\eps$ for some $j_t\in X_{t}$, which implies that $\max(\hat{X}_{t+\hat{\gamma}T})\le j_t$. Therefore, in the auction repeated for $(1+\hat{\gamma})T$ rounds with the optimizer bidding $\hat{b}_{O}^1,\dots,\hat{b}_{O}^{(1+\hat{\gamma})T}$, for any $t\in\T(r_L^*)$, if the learner's bid $\hat{j}_{t+\hat{\gamma}T}\cdot\eps$ at round $t+\hat{\gamma}T$ is $\gamma$-mean-based (i.e., $\hat{j}_{t+\hat{\gamma}T}\in \hat{X}_{t+\hat{\gamma}T}$), then it must hold that $\hat{j}_{t+\hat{\gamma}T}\cdot\eps\le j_t\cdot\eps=b_{L}^t$. Moreover, we can show that our choice of bids $\hat{b}_{O}^{t+\hat{\gamma}T}$ implies that 
for each $t\in\T(r_L^*)$, the optimizer's utility by bidding $\hat{b}_{O}^{t+\hat{\gamma}T}$ against the learner's bid $\hat{j}_{t+\hat{\gamma}T}\cdot \eps$, is no less than its utility by bidding $b_{O}^t$ against the learner's bid $b_{L}^t$. Specifically, if $\min(X_t)\cdot\eps<b_O^{t}\le v_O$, then the optimizer's two bids $\hat{b}_{O}^{t+\hat{\gamma}T}$ and $b_O^{t}$ are equal, and the learner's bid $\hat{j}_{t+\hat{\gamma}T}\cdot\eps$ is no higher than $b_{L}^t=j_t\cdot\eps$ as we have shown. If $b_O^{t}>v_O$, then the optimizer's utility is at most $0$ by bidding $b_O^{t}$. If $b_O^{t}\le\min(X_t)\cdot\eps$, then the optimizer's utility by bidding $b_O^{t}$ against the learner's bid $b_L^t$ is $0$ because $b_L^t=j_t\cdot\eps\ge\min(X_t)\cdot\eps$ as $j_t\in X_{t}$ for each $t\in\T(r_L^*)$. On the other hand, if $b_O^{t}>v_O$ or $b_O^{t}\le\min(X_t)\cdot\eps$, the optimizer's utility is at least $0$ by bidding $\hat{b}_{O}^{t+\hat{\gamma}T}=0$.

We have shown that for each $t\in\T(r_L^*)$, if the learner plays $\gamma$-mean-based action at round $t+\hat{\gamma}T$, then, the optimizer's utility by bidding $\hat{b}_{O}^{(t+\hat{\gamma})T}$ at round $t+\hat{\gamma}T$ is at least the optimizer's utility by bidding $b_{O}^t$ against the learner's bid $b_{L}^t$, which is the contribution of round $t$ to $U_O(\A_O,\A_L^{(r_L^*)},\T(r_L^*))$ (recall that $U_O(\A_O,\A_L^{(r_L^*)},\T(r_L^*))$ is the optimizer's total utility by bidding $b_{O}^1,\dots,b_{O}^{T}$ against the learner's bids $b_L^1,\dots,b_L^T$ in $T$ rounds). Since the learner's randomized algorithm $\A_L^{(\tilde{r}_L)}$ is $\gamma$-mean-based, $\A_L^{(\tilde{r}_L)}$ plays $\gamma$-mean-based action with probability at least $1-\gamma N$ at each round. It follows that by bidding $\hat{b}_{O}^1,\dots,\hat{b}_{O}^{(1+\hat{\gamma})T}$ against the learner's algorithm $\A_L^{(\tilde{r}_L)}$, the optimizer's expected utility in rounds $\{t+\hat{\gamma}T\mid t\in\T(r_L^*)\}$ is at least $(1-\gamma N)U_O(\A_O,\A_L^{(r_L^*)},\T(r_L^*))\ge U_O(\A_O,\A_L^{(r_L^*)},\T(r_L^*))-\gamma N \cdot v_O\cdot T=U_O(\A_O,\A_L^{(r_L^*)},\T(r_L^*))-o(T)$. Besides, since $\hat{b}_{O}^t\le v_O$ for all $t\in[(1+\hat{\gamma})T]$, the optimizer's utility in the other rounds is non-negative, and thus, the optimizer's total expected utility in all the rounds $[(1+\hat{\gamma})T]$ is at least $U_O(\A_O,\A_L^{(r_L^*)},\T(r_L^*))-o(T)$. Therefore, it suffices to upper bound the optimizer's total expected utility when bidding $\hat{b}_{O}^1,\dots,\hat{b}_{O}^{(1+\hat{\gamma})T}$ against the learner's algorithm $\A_L^{(\tilde{r}_L)}$ in $(1+\hat{\gamma})T$ rounds.

\subsubsection*{Bid $0$ is learner's approximately best response to oblivious optimizer's bids}
Now we show that bid $0$ is an approximately best response of the learner to the optimizer's mixed strategy $\hat{\alpha}$ corresponding to the empirical distribution of the bids $\hat{b}_{O}^1,\dots,\hat{b}_{O}^{(1+\hat{\gamma})T}$, in the sense that $u_L(\hat{\alpha},0)\ge u_L(\hat{\alpha},j)-o(1)$ for any $j\in\{0,\dots,N-1\}$, where $u_L(\hat{\alpha},j)$ denotes the learner's expected utility by bidding $j\cdot\eps$ against the optimizer's mixed strategy $\hat{\alpha}$. Given the optimizer's bids $\hat{b}_{O}^1,\dots,\hat{b}_{O}^{(1+\hat{\gamma})T}$ in $(1+\hat{\gamma})T$ rounds, let $u_{t}(\tilde{r}_L)$ denote the learner's utility achieved by the randomized algorithm $\A_L^{(\tilde{r}_L)}$ at round $t\in[(1+\hat{\gamma})T]$ (note that $u_{t}(\tilde{r}_L)$ is a random variable which depends on the algorithm's randomness $\tilde{r}_L$), and let $u_{j,t}$ denote the learner's utility by bidding $j\cdot\eps$ at round $t$ for $j\in\{0,\dots,N-1\}$. Thus, the learner's expected utility achieved by the randomized algorithm $\A_L^{(\tilde{r}_L)}$ in $(1+\hat{\gamma})T$ rounds is $\E_{\tilde{r}_L}[\sum_{t=1}^{(1+\hat{\gamma})T}u_{t}(\tilde{r}_L)]$, and the learner's utility by bidding $0$ in $(1+\hat{\gamma})T$ rounds is $\sum_{t=1}^{(1+\hat{\gamma})T}u_{0,t}$. Because the learner's algorithm $\A_L^{(\tilde{r}_L)}$ is no-external-regret, if we could show $\sum_{t=1}^{(1+\hat{\gamma})T}u_{0,t}\ge \E_{\tilde{r}_L}[\sum_{t=1}^{(1+\hat{\gamma})T}u_{t}(\tilde{r}_L)]-o(T)$, then we would have that  $\sum_{t=1}^{(1+\hat{\gamma})T}u_{0,t}\ge \sum_{t=1}^{(1+\hat{\gamma})T}u_{j,t}-o(T)$ for any $j\in \{0,\dots,N-1\}$. Since $\sum_{t=1}^{(1+\hat{\gamma})T}u_{0,t}\ge \sum_{t=1}^{(1+\hat{\gamma})T}u_{j,t}-o(T)$ is equivalent to $u_L(\hat{\alpha},0)\ge u_L(\hat{\alpha},j)-o(1)$, it remains to prove $\sum_{t=1}^{(1+\hat{\gamma})T}u_{0,t}\ge \E_{\tilde{r}_L}[\sum_{t=1}^{(1+\hat{\gamma})T}u_{t}(\tilde{r}_L)]-o(T)$.

The learner's utility in the first $\hat{\gamma}T$ rounds is always $o(T)$ (since $\hat{\gamma}T=o(T)$), and thus, it suffices to show that $\sum_{t=\hat{\gamma}T+1}^{(1+\hat{\gamma})T}u_{0,t}\ge \E_{\tilde{r}_L}[\sum_{t=\hat{\gamma}T+1}^{(1+\hat{\gamma})T}u_{t}(\tilde{r}_L)]-o(T)$.
For each round $t\in\{\hat{\gamma}T+1,\dots,(1+\hat{\gamma})T\}$, we can prove that if round $t$ is $\gamma$-mean-based, then $u_{t}(\tilde{r}_L)\le u_{0,t}$. Specifically, at round $t\in\{\hat{\gamma}T+1,\dots,(1+\hat{\gamma})T\}$, if the learner's algorithm $\A_L^{(\tilde{r}_L)}$ plays a $\gamma$-mean-based action $\hat{j}_t\in\hat{X}_{t}$, then $\hat{j}_t\le\min(X_{t-\hat{\gamma}T})$ because we showed $\max(\hat{X}_t)\le\min(X_{t-\hat{\gamma}T})$ in Claim~\ref{claim:lower-mean-based-bid}. Moreover, since $\hat{b}_O^{t}=\1(\min(X_t)\cdot\eps<b_O^{t-\hat{\gamma}T}\le v_O)\cdot b_O^{t-\hat{\gamma}T}$ for $t\in\{\hat{\gamma}T+1,\dots,(1+\hat{\gamma})T\}$, either:
\begin{itemize} 
\item we have that $\hat{b}_O^{t}>\min(X_{t-\hat{\gamma}T})$, in which case both bid $\hat{j}_t\cdot\eps$ and bid $0$ lose the item and get utility $0$ for the learner; or 
\item we have that $\hat{b}_O^{t}=0$ in which case both bid $\hat{j}_t\cdot\eps$ and bid $0$ win the item for the learner, and clearly, bid $0$ gets no less utility than bid $\hat{j}_t\cdot\eps$. 
\end{itemize}
Hence, for each $t\in\{\hat{\gamma}T+1,\dots,(1+\hat{\gamma})T\}$, conditioned on that the event $E_t$ that round $t$ is $\gamma$-mean-based, it holds that $u_{t}(\tilde{r}_L)\le u_{0,t}$. Since the randomized algorithm $\A_L^{(\tilde{r}_L)}$ is $\gamma$-mean-based, each round $t$ is $\gamma$-mean-based with probability $\Pr[E_t]\ge 1-\gamma N$, and it follows that
\begin{align*}
    \E_{\tilde{r}_L}\left[\sum_{t=\hat{\gamma}T+1}^{(1+\hat{\gamma})T}u_{t}(\tilde{r}_L)\right]&=\sum_{t=\hat{\gamma}T+1}^{(1+\hat{\gamma})T}\E_{\tilde{r}_L}[u_{t}(\tilde{r}_L)]\\
    &=\sum_{t=\hat{\gamma}T+1}^{(1+\hat{\gamma})T}\Pr[E_t]\cdot\E_{\tilde{r}_L}[u_{t}(\tilde{r}_L)\mid E_t]+\Pr[\bar{E_t}]\cdot\E_{\tilde{r}_L}[u_{t}(\tilde{r}_L)\mid \bar{E_t}]\\
    &\le\sum_{t=\hat{\gamma}T+1}^{(1+\hat{\gamma})T}\Pr[E_t]\cdot u_{0,t}+\Pr[\bar{E_t}]\cdot v_L\\
    &\text{(Because $u_{t}(\tilde{r}_L)\le u_{0,t}$ conditioned on $E_t$, and $u_{t}(\tilde{r}_L)\le v_L$ regardless of $E_t$)}\\
    &=\sum_{t=\hat{\gamma}T+1}^{(1+\hat{\gamma})T} ((1-\gamma N)u_{0,t}+\gamma N \cdot v_L)\\
    &\le \left(\sum_{t=\hat{\gamma}T+1}^{(1+\hat{\gamma})T}u_{0,t}\right) + o(T)\\
    &\text{(Because $v_L\le \eps N$ by Assumption~\ref{assumption:discrete-bid} and $\gamma N^2=o(1)$).}
\end{align*}

Thus, we have shown that there exists $\delta=o(1)$ such that for any $j\in\{0,\dots,N-1\}$
\begin{equation}\label{eq:0-approx-best-resp}
    u_L(\hat{\alpha},0)\ge u_L(\hat{\alpha},j)-\delta.
\end{equation}
Before we proceed, we observe that $u_O(\hat{\alpha},0)\cdot (1+\hat{\gamma})T$ (where $u_O(\hat{\alpha},0)$ denotes the optimizer's expected utility when the learner bids $0$ against the optimizer's mixed strategy $\hat{\alpha}$) upper bounds the optimizer's expected utility in $(1+\hat{\gamma})T$ rounds by bidding $\hat{b}_{O}^1,\dots,\hat{b}_{O}^{(1+\hat{\gamma})T}$ against the learner's randomized algorithm $\A_L^{(\tilde{r}_L)}$. Specifically, $\hat{\alpha}$ by definition is exactly the empirical distribution of $\hat{b}_{O}^1,\dots,\hat{b}_{O}^{(1+\hat{\gamma})T}$, and because $\hat{b}_{O}^t\le v_O$ by definition of $\hat{b}_{O}^t$ for all $t\in[(1+\hat{\gamma})T]$, the optimizer's utility by bidding $\hat{b}_{O}^1,\dots,\hat{b}_{O}^{(1+\hat{\gamma})T}$ against $\A_L^{(\tilde{r}_L)}$ is no higher than its utility by bidding $\hat{b}_{O}^1,\dots,\hat{b}_{O}^{(1+\hat{\gamma})T}$ against bid $0$. Therefore, it remains to prove that $u_O(\hat{\alpha},0)\le V+o(1)$, which implies $u_O(\hat{\alpha},0)\cdot (1+\hat{\gamma})T\le V\cdot T+o(T)$ as $\hat{\gamma}=o(1)$.

\subsubsection*{Modifying optimizer's bids s.t.~bid $0$ is learner's true best response}
The proof would have been finished if bid $0$ is exactly a best response of the learner to the optimizer's mixed strategy $\hat{\alpha}$, because then $u_O(\hat{\alpha},0)\le V$ by definition of Stackelberg utility (see Definition~\ref{def:bayesian_stackelberg}). To finish the proof, we can modify $\hat{\alpha}$ slightly by adding a negligible probability mass on bid $0$ such that bid $0$ becomes the learner's true best response. Specifically, we let $\hat{\alpha}'$ be the optimizer's mixed strategy that bids $0$ with probability $\eta=\frac{2\delta}{\eps}=o(1)$ and bids according to mixed strategy $\hat{\alpha}$ with probability $1-\eta$. Notice that $u_L(\hat{\alpha}',j)=\eta u_L(0,j)+(1-\eta)u_L(\hat{\alpha},j)$ for all $j\in\{0,\dots,N-1\}$. Thus, for any $j>0$, we have that
\begin{align*}
    u_L(\hat{\alpha}',0)-u_L(\hat{\alpha}',j)&=\eta(u_L(0,0)-u_L(0,j))+(1-\eta)(u_L(\hat{\alpha},0)-u_L(\hat{\alpha},j))\\
    &\ge\eta(u_L(0,0)-u_L(0,j))-(1-\eta)\delta &&\text{(By Ineq.~\ref{eq:0-approx-best-resp})}\\
    &=\eta((v_L-0)-(v_L-j\cdot\eps))-(1-\eta)\delta\\
    &\ge\eta\eps-(1-\eta)\delta &&\text{(By $j>0$)}\\
    &\ge\eta\eps-\delta=\delta &&\text{(By $\eta=\frac{2\delta}{\eps}$)},
\end{align*}
which implies that bid $0$ is the learner's best response to the optimizer's mixed strategy $\hat{\alpha}'$. Finally, by definition of Stackelberg utility, we have that $V\ge u_O(\hat{\alpha}',0)=(1-\eta)u_O(\hat{\alpha},0)=u_O(\hat{\alpha},0)-o(1)$, which completes the proof.
\end{proof}

\section{Exploiting mean-based learners in Bayesian first-price auction}\label{section:bayesian}
In this section, we show that in contrast to standard first-price auctions, the optimizer is able to achieve much higher than the Stackelberg utility times the number of rounds in repeated Bayesian first-price auctions. We will use the following characterization of the optimizer's mixed strategy in the Stackelberg equilibrium of Bayesian first-price auction due to~\citet{xu2018commitment}.
\begin{lemma}[{\citet[Theorem 7]{xu2018commitment}}]\label{lem:first-price-stackelberg}
In any Bayesian first-price auction, suppose the optimizer's value is $v_O\in\I$, and there are $m$ values $v_1,\dots,v_m\in\I$ (in the strictly increasing order) in the support of the prior distribution $\D$ of the learner's value $v_L$. Then, there exists a Stackelberg equilibrium for this Bayesian auction, such that the CDF of the optimizer's bid corresponding to the optimizer's mixed strategy in the Stackelberg equilibrium, which we denote by $F$, has the following property:
\begin{itemize}
    \item[(i)] There exist $b_1=0$ and $b_2,\dots, b_{m+1}\in[0,v_O]\cap\I$ (in the non-decreasing order) such that $F$ satisfies
    \begin{equation}\label{eq:first-price-stackelberg-F}
    F(x)=
    \frac{(v_i-b_i)F(b_i)}{v_i-x} \textrm{ for } x\in(b_i,b_{i+1}]\cap\I,\,
    \forall i\in[m],
    \end{equation}
    \item[(ii)] and moreover, for each $i\in[m]$, when the learner's value $v_L$ is $v_i$, bid $b_i$ is the learner's best response to the optimizer's mixed strategy.
\end{itemize}
\end{lemma}

\begin{theorem}\label{thm:exploit-bayesian-mean-based}
There exists a Bayesian first-price auction repeated for $T$ rounds such that the optimizer's optimal expected utility in $T$ rounds against any $\gamma$-mean-based learner with $\gamma N=o(1)$ is at least $\Omega(\log(N\cdot\eps))\cdot V\cdot T$, where $V$ is the Stackelberg utility of this instance (and $N,\eps$ are defined in Assumption~\ref{assumption:discrete-bid}).
\end{theorem}
We first give the construction of the instance of Bayesian first-price auction for Theorem~\ref{thm:exploit-bayesian-mean-based} and sketch the high-level idea, and then we prove Theorem~\ref{thm:exploit-bayesian-mean-based}.
\begin{proof}[Construction of the instance]\renewcommand{\qedsymbol}{}
The optimizer's value is $v_{O}=1$. The prior distribution $\D$ of the learner's value $v_L$ has support $\{v_i:=2^i\mid i\in[m]\}$ for $m=\floor{\log(N\cdot\eps)}$, and the probabilities are as follows: $\P[v_L=2^i]=\frac{1}{2^{m-i}}$ for $i\in[m-1]$, and $\P[v_L=2^m]=1-\sum_{i\in[m-1]}\frac{1}{2^{m-i}}=\frac{1}{2^{m-1}}$. (We assume w.l.o.g.~that $N,\eps$ in Assumption~\ref{assumption:discrete-bid} are chosen such that $2^{i}\in\I$ for all $i\in\{0,\dots,m\}$.)
\end{proof}
\begin{proof}[High-level idea]\renewcommand{\qedsymbol}{}
In the instance we constructed above, the learner's possible values ordered from low to high are geometrically increasing. The optimizer's value is lower than the lowest possible value of the learner, and hence, the optimal mixed strategy of the optimizer should incentivize the learner to bid low such that the optimizer itself can get some utility. However, when the learner has high value, it is willing to bid higher (compared to when it has low value) in order to maximize its utility, and hence, it is harder for the optimizer to incentivize. Therefore, from the optimizer's perspective, its optimal mixed strategy against a learner with high value should be significantly less ambitious (i.e., bid low with higher probability) compared to the optimal mixed strategy against a learner with low value, but such mixed strategy is suboptimal against a learner with low value. Thus, if the optimizer only has one shot against a learner whose value is sampled from $\D$ but unknown to the optimizer, it must commit to a single mixed strategy, and as a result, the optimizer cannot simultaneously get the best utility it can get for all the possible learner's values -- in particular, we can show that in the Stackelberg equilibrium, the optimizer essentially can only get the best utility it can get for one specific learner's value.

On the other hand, when the auction is repeated for $T$ rounds, the optimizer can change its strategy over time. Specifically, the optimizer can divide the time horizon into roughly $m$ phases with decreasing lengths but increasingly ambitious strategies: In the first phase, the cumulative distribution of the optimizer's bids will resemble the optimal mixed strategy against a learner with highest possible value ($2^m$ in our instance), and similarly, in the $i$-th phase for $i\ge 2$, the cumulative distribution of the optimizer's bids (in the first $i$ phases) will resemble the optimal mixed strategy against a learner with the $i$-th highest possible value. We can show that with this dynamic strategy, in the first $i$ phases for $i\in[m]$, the optimizer can get a constant fraction of the best utility it can get (across all rounds) for the $i$-th highest possible learner's value.
Therefore, using this dynamic strategy, overall the optimizer is able to get a constant fraction of the best utility it can get for all the possible learner's values, which implies the gap.
\end{proof}
\begin{proof}[Proof of Theorem~\ref{thm:exploit-bayesian-mean-based}]
We begin the proof by establishing an upper bound of the Stackelberg utility $V\le\frac{1}{2^{m-3}}$ for the instance we constructed, and then we show that when the instance is repeated for $T$ rounds, the optimizer can achieve expected total utility $\frac{\Omega(m)\cdot T}{2^m}$ against any $\gamma$-mean-based learner with $\gamma N=o(1)$.

\subsubsection*{Upper bounding the Stackelberg utility}
We consider a Stackelberg equilibrium of our instance that satisfies the characterization given in Lemma~\ref{lem:first-price-stackelberg}. We let $F$ denote the CDF of the optimizer's bid as in Lemma~\ref{lem:first-price-stackelberg}, and we let $j_i$ be such that $j_i\cdot\eps=b_i$ for each $i\in[m]$ (where $b_i$'s are given in Lemma~\ref{lem:first-price-stackelberg}). Moreover, we let $f(j\cdot \eps):=F(j\cdot\eps)-F((j-1)\cdot\eps)$ for all $j\in[N]$ and $f(0):=F(0)$. First, we derive the optimizer's expected utility (i.e., the Stackelberg utility $V$) in the following equation,
\begin{equation}\label{eq:separation-stackelberg-utility}
V=\sum_{i=1}^m \sum_{j=j_i+1}^{j_{i+1}} f(j\cdot \eps)\cdot\P[v_L\le 2^i]\cdot (1-j\cdot\eps).
\end{equation}
To see this, note that $b_{k}=j_{k}\cdot\eps$ is the best response of the learner if its value is $v_{k}=2^{k}$ for each $k\in[m]$ (by Lemma~\ref{lem:first-price-stackelberg}), and hence, when the optimizer bids $j\cdot \eps$ for $j\in\{j_i+1,\dots,j_{i+1}\}$, it wins the item (and gets utility $v_O-j\cdot\eps=1-j\cdot\eps$) if and only if the learner's value is $2^{k}$ for some $k\in[i]$. 

We upper bound each term in the summands in Eq.~\eqref{eq:separation-stackelberg-utility}. For $\P[v_L\le 2^i]$, we derive a simple upper bound:
\begin{equation}\label{eq:separation-prior-upper-bound}
\forall i\in[m],\,\P[v_L\le 2^i]\le\sum_{j=1}^i\frac{1}{2^{m-j}}=\frac{1}{2^{m-i-1}}-\frac{1}{2^{m-1}}\le\frac{1}{2^{m-i-1}}.
\end{equation}
For $f(j\cdot\eps)$ with $j\in\{j_i+1,\dots,j_{i+1}\}$ for $i\in[m]$, we have that
\begin{align}
    f(j\cdot\eps)&=\frac{(v_i-b_i)F(b_i)}{v_i-j\cdot\eps}-\frac{(v_i-b_i)F(b_i)}{v_i-(j-1)\cdot\eps} &&\text{(By Eq.~\eqref{eq:first-price-stackelberg-F})}\nonumber\\
    &=\frac{(v_i-b_i)F(b_i)\cdot\eps}{(v_i-j\cdot\eps)(v_i-(j-1)\cdot\eps)}\nonumber\\
    &\le\frac{v_i\cdot\eps}{(v_i-j\cdot\eps)(v_i-(j-1)\cdot\eps)}&&\text{(Since $v_i-b_i\le v_i$ and $F(b_i)\le1$)}\nonumber\\
    &\le\frac{v_i\cdot\eps}{(v_i-v_O)^2} &&\text{(Since $j\cdot\eps\le j_{i+1}\cdot\eps=b_{i+1}$ and by Lemma~\ref{lem:first-price-stackelberg} $b_{i+1}\le v_O$)}\nonumber\\
    &=\frac{2^{i}\cdot\eps}{(2^{i}-1)^2}\le\frac{\eps}{2^{i-2}} &&\text{(Since $2^i-1\ge2^{i-1}$ for $i\ge1$)}.\label{eq:separation-Pj-upper-bound}
\end{align}
Now we upper bound $V$ as follows,
\begin{align*}
    V &= \sum_{i=1}^m \sum_{j=j_i+1}^{j_{i+1}} f(j\cdot\eps)\cdot\P[v_L\le 2^i]\cdot (1-j\cdot\eps) &&\text{(By Eq.~\eqref{eq:separation-stackelberg-utility})}\\
    &\le \sum_{i=1}^m \sum_{j=j_i+1}^{j_{i+1}} \frac{\eps}{2^{i-2}}\cdot\frac{1}{2^{m-i-1}}\cdot 1 &&\text{(By Ineq.~\eqref{eq:separation-prior-upper-bound} and~\eqref{eq:separation-Pj-upper-bound})}\\
    &=\sum_{i=1}^m \sum_{j=j_i+1}^{j_{i+1}} \frac{\eps}{2^{m-3}}=\frac{(j_{m+1}-j_1)\cdot\eps}{2^{m-3}}=\frac{b_{m+1}-b_1}{2^{m-3}}\le\frac{1}{2^{m-3}},
\end{align*}
where the final inequaility is because $b_1=0$ and $b_{m+1}\le v_O=1$ by Lemma~\ref{lem:first-price-stackelberg}.

\subsubsection*{Designing an oblivious strategy for optimizer in repeated auctions}
Given any $\gamma$ such that $\gamma N =o(1)$, we construct an oblivious sequence of bids in $T$ rounds for the optimizer, and later we will show that the optimizer's expected utility by bidding this sequence of bids against any $\gamma$-mean-based learner in $T$ rounds is at least  $\frac{\Omega(m)\cdot T}{2^m}$. Let $F_{m}(x):=\frac{v_1}{2(v_1-x)}$, and then let $F_{i}(x):=\frac{v_{m-i+1}}{2(v_{m-i+1}-x)}$ for all $i\in[m-1]$, and let $F_0(x):=0$. The sequence we construct has $m+1$ phases:
\begin{itemize}
    \item[$(0)$] In the zeroth phase (the first $\gamma' T$ rounds for $\gamma':=\frac{2\gamma}{\eps}=o(1)$), the optimizer always bids $0$.
    \item[$(i)$] For any $i\in[m]$, in the $i$-th phase, the optimizer's bids are in the increasing order, and for all $x\in[0,1]\cap\I$, there are exactly $(1-\gamma')T\cdot (F_i(x)-F_{i-1}(x))$ bids less than or equal to $x$ in the $i$-th phase.
\end{itemize}
Note that the total number of bids from phase $(1)$ to $(m)$ is $(1-\gamma')T\cdot F_1(1)+\sum_{i=2}^m(1-\gamma')T\cdot (F_i(1)-F_{i-1}(1))=(1-\gamma')T\cdot F_m(1)=(1-\gamma')T$, and hence, the total number of bids in $m+1$ phases is indeed $T$. Moreover, the first phase is well-defined because $F_1(x)-F_0(x)$ is obviously a non-decreasing function over $[0,1]$, and for $i\in\{2,\dots,m\}$, the $i$-th phase is well-defined because we can show that $F_i(x)-F_{i-1}(x)=\frac{v_{m-i+1}}{2(v_{m-i+1}-x)}-\frac{v_{m-i+2}}{2(v_{m-i+2}-x)}$ is non-decreasing over $[0,1]$ as follows
\begin{align*}
    \frac{dF_i}{dx}-\frac{dF_{i-1}}{dx}&=\frac{v_{m-i+1}}{2(v_{m-i+1}-x)^2}-\frac{v_{m-i+2}}{2(v_{m-i+2}-x)^2}\\
    &\ge \frac{v_{m-i+1}}{2v_{m-i+1}^2}-\frac{v_{m-i+2}}{2(v_{m-i+2}-1)^2} &&\text{(Since $0\le x\le1$)}\\
    &= \frac{1}{2^{m-i+2}}-\frac{2^{m-i+1}}{(2^{m-i+2}-1)^2}\ge 0,
\end{align*}
where the last inequality is because for any $\ell\ge1$, it holds that $(2^{\ell+1}-1)^2=2^{2\ell+2}+1-2^{\ell+2}\ge2^{2\ell+2}-2^{\ell+2}\ge2^{2\ell+2}-2^{2\ell+1}=2^{2\ell+1}$, and in particular $(2^{m-i+2}-1)^2\ge2^{2(m-i+1)+1}$ as $m-i+1\ge1$.
Furthermore, it is also worth noting that the optimizer's bids are at most $1=v_O$ in the oblivious sequence we constructed, and hence, the optimizer's utility in each round is always non-negative.

\subsubsection*{Lower bounding optimizer's expected utility in repeated auctions}
Next, we lower bound the optimizer's expected utility when it uses the oblivious sequence of bids we constructed above against any $\gamma$-mean-based learner. The key property we will use is Claim~\ref{claim:separation-phase-behavior}.
\begin{claim}\label{claim:separation-phase-behavior}
For any $i\in[m]$, 
at any round of the $i$-th phase, bid $0$ is the only $\gamma$-mean-based bid for the learner if $v_L\le v_{m-i+1}$.
\end{claim}
\begin{proof}
First, notice that for all $x\in[0,1]\cap\I$, the total number of optimizer's bids that are less than or equal to $x$ in the first $i+1$ phases is $\sum_{i'=1}^i(1-\gamma')T\cdot (F_{i'}(x)-F_{i'-1}(x))=(1-\gamma')T\cdot F_{i}(x)$.

Because the optimizer's bids in the $i$-th phase are in the increasing order, by round $(1-\gamma')T\cdot (F_i(0)-F_{i-1}(0))$ of the $i$-th phase (to be clear, round $k$ of the $i$-th phase refers to the $k$-th round in the $i$-th phase), the optimizer has submitted all the zero bids of the $i$-th phase. Thus, at any round $t_{i}\in\{(1-\gamma')T\cdot (F_i(0)-F_{i-1}(0))+1,\dots,(1-\gamma')T\cdot (F_i(1)-F_{i-1}(1))\}$ of the $i$-th phase, the learner has seen exactly $\gamma'T+(1-\gamma')T\cdot F_{i}(0)$ zero bids in the past (and for any $x\in(0,1]\cap\I$, the learner has seen at most $\gamma'T+(1-\gamma')T\cdot F_{i}(x)$ bids that are less than or equal to $x$, because there are only so many bids less than or equal to $x$ in the first $i+1$ phases).

Let $u_{t_{i}}(x|v_{L})$ denote the cumulative utility of the learner with value $v_L$ by bidding $x$ from the zeroth phase until the $t_i$-th round of the $i$-th phase.
To prove the claim, it suffices to show that if $v_L\le v_{m-i+1}$, then $u_{t_{i}}(0|v_L)>u_{t_{i}}(x|v_L)+\gamma T$ for all $x\in(0,1]\cap\I$.
To this end, notice that $u_{t_{i}}(0|v_L)=(\gamma'T+(1-\gamma')T\cdot F_{i}(0))\cdot v_L$, and for $x\in(0,1]\cap\I$, $u_{t_{i}}(x|v_L)\le(\gamma'T+(1-\gamma')T\cdot F_{i}(x))\cdot(v_L-x)$. Hence, for $v_L\le v_{m-i+1}$ and $x\in(0,1]\cap\I$, we have that
\begin{align*}
    &u_{t_{i}}(0|v_L)-u_{t_{i}}(x|v_L)\\
    \ge&\gamma'T\cdot x+(1-\gamma')T\cdot(F_i(0)\cdot v_L-F_i(x)\cdot (v_L-x))\\
    =&\gamma'T\cdot x+(1-\gamma')T\cdot\left(\frac{v_{m-i+1} \cdot v_L}{2v_{m-i+1}}-\frac{v_{m-i+1}(v_L-x)}{2(v_{m-i+1}-x)}\right) &&\text{(By definition of $F_i$)}\\
    =&\gamma'T\cdot x+(1-\gamma')T\cdot \frac{v_{m-i+1} (v_{m-i+1}-v_L)x}{2v_{m-i+1}(v_{m-i+1}-x)}\\
    \ge&\gamma'T\cdot x &&\text{(By $v_{m-i+1}-v_L\ge0$)}\\
    \ge&\gamma'T\cdot \eps &&\text{(By $x>0$ and $x\in\I$)}\\
    >&\gamma T &&\text{(By $\gamma'=\frac{2\gamma}{\eps}$)}.
\end{align*}
\end{proof}
We are ready to lower bound the optimizer's expected utility. Let $P_i(x)$ denote the fraction of bids that are equal to $x$ in the $i$-th phase for $i\in[m]$. For $x\in(\frac{1}{2},1]\cap\I$, we derive that
\begin{align}\label{eq:separation-P_i-lower-bound}
    P_i(x)&=(F_i(x)-F_{i-1}(x))-(F_i(x-\eps)-F_{i-1}(x-\eps))\nonumber\\
    &=\left(\frac{v_{m-i+1}}{2(v_{m-i+1}-x)}-\frac{v_{m-i+2}}{2(v_{m-i+2}-x)}\right)-\left(\frac{v_{m-i+1}}{2(v_{m-i+1}-x+\eps)}-\frac{v_{m-i+2}}{2(v_{m-i+2}-x+\eps)}\right)\nonumber\\
    &=\frac{v_{m-i+1}\cdot\eps}{2(v_{m-i+1}-x)(v_{m-i+1}-x+\eps)}-\frac{v_{m-i+2}\cdot\eps}{2(v_{m-i+2}-x)(v_{m-i+2}-x+\eps)}\nonumber\\
    &=\frac{2^{m-i}\cdot\eps}{(2^{m-i+1}-x)(2^{m-i+1}-x+\eps)}-\frac{2^{m-i+1}\cdot\eps}{(2^{m-i+2}-x)(2^{m-i+2}-x+\eps)}\nonumber\\
    &\ge\frac{2^{m-i}\cdot\eps}{(2^{m-i+1}-\frac{1}{2})^2}-\frac{2^{m-i+1}\cdot\eps}{(2^{m-i+2}-1)^2} \qquad\qquad\qquad\qquad\qquad\quad\text{(Since $\frac{1}{2}+\eps\le x\le 1$)}\nonumber\\
    &=\frac{(2^{m-i+2}-2^{m-i+1})\cdot\eps}{(2^{m-i+2}-1)^2}=\frac{2^{m-i+1}\cdot\eps}{(2^{m-i+2}-1)^2}\ge\frac{2^{m-i+1}\cdot\eps}{(2^{m-i+2})^2}=\frac{\eps}{2^{m-i+3}}.
\end{align}
Moreover, we can show that for $i\in[m]$, in the $i$-th phase, the optimizer's expected utility is lower bounded by $\sum_{x\in(0,1]\cap\I}\Pr[v_L\le v_{m-i+1}]\cdot(1-o(1))\cdot P_i(x)\cdot T\cdot (1-x)$. Specifically, in the $i$-th phase, because the optimizer's bids are in the increasing order, and there are $(1-\gamma')T\cdot(F_i(0)-F_{i-1}(0))$ zero bids in the $i$-th phase, the optimizer bids $x\in(0,1]\cap\I$ only after round $(1-\gamma')T\cdot(F_i(0)-F_{i-1}(0))$ of the $i$-th phase.
Therefore, by Claim~\ref{claim:separation-phase-behavior}, whenever the optimizer bids $x\in(0,1]\cap\I$ in the $i$-th phase, conditioned on the learner's value $v_L\le v_{m-i+1}$, the only $\gamma$-mean-based bid for the learner is $0$, and hence, the learner bids $0$ with probability at least $1-\gamma N=1-o(1)$ (since the learner is $\gamma$-mean-based with $\gamma N=o(1)$), and the optimizer wins the item and gets utility $v_O-x=1-x$.

Summing over all $i\in[m]$, we get the following lower bound of the optimizer's expected utility,
\begin{align*}
    &\sum_{i\in[m]}\sum_{x\in(0,1]\cap\I}\Pr[v_L\le v_{m-i+1}]\cdot(1-o(1))\cdot P_i(x)\cdot T\cdot (1-x)\\
    \ge&\sum_{i\in\{2,\dots,m\}}\sum_{x\in(\frac{1}{2},\frac{3}{4}]\cap\I}\Pr[v_L=v_{m-i+1}]\cdot(1-o(1))\cdot P_i(x)\cdot T\cdot (1-x)\\
    =&\sum_{i\in\{2,\dots,m\}}\sum_{x\in(\frac{1}{2},\frac{3}{4}]\cap\I}\frac{1}{2^{i-1}}\cdot(1-o(1))\cdot P_i(x)\cdot T\cdot (1-x)\\
    \ge&\sum_{i\in\{2,\dots,m\}}\sum_{x\in(\frac{1}{2},\frac{3}{4}]\cap\I}\frac{1}{2^{i-1}}\cdot(1-o(1))\cdot \frac{\eps T}{2^{m-i+3}}\cdot (1-x) &&\text{(By Ineq.~\eqref{eq:separation-P_i-lower-bound})}\\
    =&\frac{(m-1)\cdot(1-o(1))\cdot\eps T}{2^{m+2}}\cdot\sum_{x\in(\frac{1}{2},\frac{3}{4}]\cap\I} (1-x)\\
    =&\frac{\Omega(m)\cdot T}{2^m},
\end{align*}
where the last equality is because the size of $(\frac{1}{2},\frac{3}{4}]\cap\I$ is roughly $(\frac{3}{4}-\frac{1}{2})/\eps=\frac{1}{4\eps}$, and $1-x\ge\frac{1}{4}$ holds for all $x\in(\frac{1}{2},\frac{3}{4}]$.
\end{proof}

\section{Minimizing polytope swap regret in Bayesian first-price auction}\label{section:minimize-polytope-swap}
In Theorem~\ref{thm:exploit-bayesian-mean-based}, we showed that there is an instance of repeated Bayesian first-price auction where the optimizer achieves much higher utility than the Stackelberg utility times the number of rounds against mean-based learners. Are there stronger ``no-regret'' learners that can cap the optimizer's utility at the Stackelberg utility? \citet{mansour2022strategizing} gave a positive answer to this question for general Bayesian games -- they showed that the optimizer cannot achieve more than the Stackelberg utility times the number of rounds against any no-polytope-swap-regret learner in any repeated Bayesian game (see Lemma \ref{lem:mansour-theorem}).
Moreover, they observed that one can minimize $\S$-polytope swap regret using any no-swap-regret\footnote{Swap regret is a well-studied notion of regret in standard full-information games~\citep{blum2007external}, which is equivalent to polytope swap regret in standard games.} algorithm (e.g.,~\citet{blum2007external}) by treating every pure strategy $f\in\S$ as a ``single action'' in a standard (i.e., full-information) game.
\begin{lemma}[{\citet[Theorem 9]{mansour2022strategizing}}]\label{lem:mansour-alg}
If there is an algorithm that incurs swap regret at most $r(n,T)$ and has runtime $t(n)$ per round in any repeated standard game with $n$ actions of the learner and $T$ rounds, then there is an algorithm that guarantees $\Reg_{\spoly{\S}}\le r(|\S|,T)$ and 
has runtime $t(|\S|)$ per round for any repeated Bayesian game $\G(M,N,C,\D,u_O,u_L)$ with $T$ rounds and any $\S\subseteq[N]^{[C]}$.
\end{lemma}
By combining Lemma~\ref{lem:mansour-alg} (for $\S=[N]^{[C]}$) with the no-swap-regret algorithm by~\citet{blum2007external} (which incurs swap regret $O(\sqrt{n\log n\cdot T})$ and has runtime $n^{O(1)}$),~\citet{mansour2022strategizing} gave an algorithm that guarantees $\Reg_{\poly}=O(\sqrt{N^CC\log N\cdot T})$ for general Bayesian games. Can we improve this algorithm's regret bound and runtime for specific Bayesian games with additional structure (e.g.,~Bayesian first-price auction in our case)? In this section, we give a slight improvement of this algorithm for Bayesian first-price auctions. Our improvement is based on the simple observation that in order to minimize polytope swap regret, it suffices to restrict the learner's attention to a subset of pure strategies that ``covers'' all the potential best responses.
Although this observation is simple in hindsight, 
we think it is worth formalizing because of its applicability. To this end, we start by defining \emph{best response cover}.
\begin{definition}\label{def:best-resp-cover}
In a Bayesian game $\G(M,N,C,\D,u_O,u_L)$, we say that a set of pure strategies $\br(\G)\subseteq[N]^{[C]}$ is a best response cover for $\G$, if for any mixed strategy $\alpha\in\Delta([M])$ of the optimizer, there exists some pure strategy $f\in\br(\G)$ of the learner which is a best response to $\alpha$.
\end{definition}
Given Definition~\ref{def:best-resp-cover}, our observation mentioned before can be formalized as follows.
\begin{lemma}\label{lem:best-resp-polytope-regret}
For any Bayesian game $\G(M,N,C,\D,u_O,u_L)$ and any best response cover $\br(\G)$ for $\G$, it holds that $\Reg_{\poly}=\Reg_{\spoly{\br(\G)}}$.
\end{lemma}
\begin{proof}
Because $\br(\G)\subseteq [N]^{[C]}$, $\Reg_{\poly}\ge\Reg_{\spoly{\br(\G)}}$ obviously holds by Definition~\ref{def:S-poly-regret}, and hence, it suffices to show that $\Reg_{\spoly{\br(\G)}}\le\Reg_{\poly}$. Given any $i_t\in[M]$ and $\rho_t\in\Delta([N]^{[C]})$ for $t\in[T]$, we denote $\vec{i}:=(i_t)_{t\in[T]}$ and $\vec{\rho}:=(\rho_t)_{t\in[T]}$ for simplicity. Then, for each $f\in[N]^{[C]}$ such that $\rho_{t,f}>0$ for some $t\in[T]$, we let $\alpha^{(\vec{i},\vec{\rho},f)}\in\Delta([M])$ denote the optimizer's mixed strategy that with probability $\frac{\rho_{t,f}}{\sum_{t'\in[T]} \rho_{t',f}}$ plays action $i_t$ for each $t\in[T]$, and we let $\pi^{(\vec{i},\vec{\rho})}(f)\in\br(\G)$ be an arbitrary best response of the learner to the optimizer's mixed strategy $\alpha^{(\vec{i},\vec{\rho},f)}$ in $\br(\G)$. Moreover, for any $f'\in[N]^{[C]}$ for which $\pi^{(\vec{i},\vec{\rho})}(f')$ is still unspecified, we let $\pi^{(\vec{i},\vec{\rho})}(f')$ be an arbitrary pure strategy in $\br(\G)$. Therefore, $\pi^{(\vec{i},\vec{\rho})}$ is a well-defined map from $[N]^{[C]}$ to $\br(\G)$ which depends on the $\vec{i}$ and $\vec{\rho}$. We derive that for any $\pi:[N]^{[C]}\to[N]^{[C]}$,
\begin{align}
&\sum_{t\in[T]}\sum_{f\in[N]^{[C]}}\rho_{t,f}\cdot u_L(i_t,\pi(f))\nonumber\\
=&\sum_{f\in[N]}\sum_{t\in[T]}\rho_{t,f}\cdot u_L(i_t,\pi(f))\label{eq:line1}\\
=&\sum\limits_{\substack{f\in[N]^{[C]} \\ \textrm{ s.t.~$\exists t\in[T],\,\rho_{t,f}>0$}}} \sum_{t'\in[T]} \rho_{t',f}\cdot \sum_{t\in[T]}\frac{\rho_{t,f}}{\sum_{t'\in[T]} \rho_{t',f}}\cdot u_L(i_t,\pi(f))\\
=&\sum\limits_{\substack{f\in[N]^{[C]} \\ \textrm{ s.t.~$\exists t\in[T],\,\rho_{t,f}>0$}}} \sum_{t'\in[T]} \rho_{t',f}\cdot \sum_{t\in[T]} u_L(\alpha^{(\vec{i},\vec{\rho},f)},\pi(f))&&\text{(By defintion of $\alpha^{(\vec{i},\vec{\rho},f)}$)}\label{eq:line3}\\
\le&\sum\limits_{\substack{f\in[N]^{[C]} \\ \textrm{ s.t.~$\exists t\in[T],\,\rho_{t,f}>0$}}} \sum_{t'\in[T]} \rho_{t',f}\cdot \sum_{t\in[T]} u_L(\alpha^{(\vec{i},\vec{\rho},f)},\pi^{(\vec{i},\vec{\rho})}(f)) &&\text{(By defintion of $\pi^{(\vec{i},\vec{\rho})}(f)$)}\nonumber\\
=&\sum_{t\in[T]}\sum_{f\in[N]^{[C]}}\rho_{t,f}\cdot u_L(i_t,\pi^{(\vec{i},\vec{\rho})}(f)) &&\text{(Analogous to Eq.~(\ref{eq:line1}-\ref{eq:line3}))}\nonumber.
\end{align}
Therefore, we have that
\begin{align*}
    \Reg_{\poly}&=\min_{\forall t\in[T],\,\rho_t\in\class{\beta_t}}\max_{\pi:[N]^{[C]}\to[N]^{[C]}}\sum_{t\in[T]}\sum_{f\in[N]^{[C]}}\rho_{t,f}\cdot u_L(i_t,\pi(f))-u_L(i_t,\beta_t) \\
    &\le\min_{\forall t\in[T],\,\rho_t\in\class{\beta_t}}\sum_{t\in[T]}\sum_{f\in[N]^{[C]}}\rho_{t,f}\cdot u_L(i_t,\pi^{(\vec{i},\vec{\rho})}(f))-u_L(i_t,\beta_t)\\
    &\le\min_{\forall t\in[T],\,\rho_t\in\class{\beta_t}}\max_{\pi:[N]^{[C]}\to\br(\G)}\sum_{t\in[T]}\sum_{f\in[N]^{[C]}}\rho_{t,f}\cdot u_L(i_t,\pi(f))-u_L(i_t,\beta_t)=\Reg_{\spoly{\br(\G)}}.
\end{align*}
\end{proof}
By Lemma~\ref{lem:best-resp-polytope-regret}, in order to minimize $\Reg_{\poly}$, it suffices to minimize $\Reg_{\spoly{\br(\G)}}$ instead. Hence, by applying Lemma~\ref{lem:mansour-alg} to $\br(\G)$ instead of $[N]^{[C]}$, we get the following proposition.
\begin{proposition}\label{prop:better-alg}
If there is an algorithm that incurs swap regret at most $r(n,T)$ and has runtime $t(n)$ per round in any repeated standard game with $n$ actions of the learner and $T$ rounds, then there is an algorithm that guarantees $\Reg_{\poly}\le r(|\br(\G)|,T)$ and 
has runtime $t(|\br(\G)|)$ per round for any repeated Bayesian game $\G(M,N,C,\D,u_O,u_L)$ with $T$ rounds.
\end{proposition}
By Proposition~\ref{prop:better-alg}, we can get an algorithm with better regret bound and runtime for a specific Bayesian game $\G$ if we can find a best response cover $\br(\G)$ for $\G$ such that $|\br(\G)|$ is substantially less than $N^C$. As we now illustrate, this is indeed the case for Bayesian first-price auctions, and the key property we use is Lemma~\ref{lem:first-price-monotone-best-resp}. which shows that the potential best responses in Bayesian first-price auctions have certain \emph{monotone} structure.
\begin{lemma}\label{lem:first-price-monotone-best-resp}
Consider any $v_1,v_2\in\I$ such that $v_1\le v_2$. In a Bayesian first-price auction, given an arbitrary mixed strategy $\alpha$ of the optimizer, suppose that bid $b_1\in I$ is a best response to $\alpha$ when the learner has value $v_1$. Then, there exists some $b_2\in I$ with $b_2\ge b_1$ such that bid $b_2$ is a best response to $\alpha$ when the learner has value $v_2$.
\end{lemma}
\begin{proof}
Let $F$ denote the CDF of the optimizer's bid corresponding to the mixed strategy $\alpha$. Since $b_1$ is a best response against the optimizer's mixed strategy $\alpha$ for the learner with value $v_1$, we have that $(v_1-b_1)F(b_1)\ge (v_1-b)F(b)$ for any $b\in\I$, which is equivalent to $v_1(F(b_1)-F(b))\ge b_1F(b_1)-bF(b)$ after rearrangement. Because $v_2\ge v_1$ and $F(b_1)\ge F(b)$ for any $b\le b_1$, it holds for any $b\le b_1$ that $v_2(F(b_1)-F(b))\ge v_1(F(b_1)-F(b))\ge b_1F(b_1)-bF(b)$, which is equivalent to $(v_2-b_1)F(b_1)\ge (v_2-b)F(b)$ after rearrangement. Thus, if the learner has value $v_2$, then the learner's expected utility by bidding $b_1$ against $\alpha$ is no less than that by bidding any $b\le b_1$, which implies that there exists a best response $b_2\ge b_1$ to $\alpha$ for the learner with value $v_2$.
\end{proof}
\begin{proposition}\label{prop:first-price-better-alg}
For any Bayesian first-price auction $\G_{\fpa}$ under Assumption~\ref{assumption:discrete-bid}, there exists a best response cover $\br(\G_{\fpa})$ such that $|\br(\G_{\fpa})|\le 4^{N}$.
\end{proposition}
\begin{proof}
Recall that by Assumption~\ref{assumption:discrete-bid}, there are $N$ possible actions (bids) $\I$ and $N$ possible contexts (values) $\I$ for the learner, and a pure strategy $f:\I\to\I$ specifies a bid $f(v_L)\in\I$ for each possible learner's value $v_L\in\I$. Given any mixed strategy $\alpha$ of the optimizer, Lemma~\ref{lem:first-price-monotone-best-resp} implies that there exists a best response $f^{\alpha}:\I\to\I$ of the learner to $\alpha$ which satisfies a nice monotone non-decreasing property: the learner's bid is non-decreasing with respect to the learner's value, i.e., $f^{\alpha}$ is a monotone non-decreasing function. Therefore, we can let $\br(\G_{\fpa})$ be the set of all the pure strategies $f:\I\to\I$ that are monotone non-decreasing functions, and it follows that $\br(\G_{\fpa})$ is a best response cover for $\G_{\fpa}$.

It is folklore that the number of monotone non-decreasing functions from an ordered set of size $N$ (such as $\I$) to itself is $\binom{2N-1}{N}$ (see e.g.,~\citep{ballsintobins}). The proof finishes since $\binom{2N-1}{N}\le\sum_{i=0}^{2N-1}\binom{2N}{i}=2^{2N-1}<4^N$.
\end{proof}
Thus, by combining Proposition~\ref{prop:better-alg} and Proposition~\ref{prop:first-price-better-alg}  with the no-swap-regret algorithm by~\citet{blum2007external} (which incurs swap regret $O(\sqrt{n\log n\cdot T})$ and has runtime $n^{O(1)}$), we have the following corollary:
\begin{corollary}\label{cor:first-price-better-alg}
For any Bayesian first-price auction $\G_{\fpa}$ under Assumption~\ref{assumption:discrete-bid}, there is an algorithm for the learner which runs in time $2^{\Theta(N)}$ and guarantees $\Reg_{\poly}=O(\sqrt{4^{N}N\cdot T})$.
\end{corollary}
In comparison, na\"ively combining Lemma~\ref{lem:mansour-alg} (for $\S=[N]^{[N]}$) with the no-swap-regret algorithm by~\citet{blum2007external} only guarantees $\Reg_{\poly}=O(\sqrt{N^{N+1}\log N\cdot T})$ and runtime $N^{\Theta(N)}$ in the worst case.

\section{Exploiting high-polytope-swap-regret learners in Bayesian game}\label{section:exploitable-polytope-swap}
As we mentioned in Section~\ref{section:preliminaries},~\citet{mansour2022strategizing} showed that no-polytope-swap-regret algorithm is sufficient to cap the optimizer's utility in any Bayesian game (Lemma~\ref{lem:mansour-theorem}). A natural open question posed by~\citet[Question 1]{mansour2022strategizing} is whether no-polytope-swap-regret algorithm is also necessary to cap the optimizer's utility. Let us explain their question formally. As common in the no-regret learning literature,~\citet{mansour2022strategizing} consider\footnote{To be precise, in their formal setup, the learner takes ``reward vectors'' as input not necessarily in the context of games, but obviously we can always associate ``reward vectors'' with the learner's utilities in certain game. The reason why we did not directly adopt their setup in the preliminary is that our results in the previous sections hold for algorithms that are not necessarily reward-based.} the online contextual learning algorithms that maximize the expected utility based on the ``reward vectors'' in the past:

\begin{definition}[Reward-based learner] \label{def:reward-based} 
In a Bayesian game $\G(M,N,C,\D,u_O,u_L)$ repeated for $T$ rounds, suppose that at each round $t\in[T]$, the optimizer plays action $i_{t}\in[M]$, and we call $\vec{r}_{t}:=(u_L(i_t,j,c))_{j\in[N],c\in[C]}$ the reward vector at round $t$. A reward-based algorithm determines the learner's mixed strategies $\beta_t$'s based on the reward vectors $\vec{r}_{t}$'s, in order to maximize the learner's expected utility $\sum_{t\in[T]}\sum_{c\in[C]}p_c\cdot u_L(i_t,\beta_t,c)$. To be precise, an algorithm $\A$ is reward-based if for all $t\in[T]$, the mixed strategy $\beta_t$ used by $\A$ at round $t$ is a deterministic function of the prior distribution $\D$ and the reward vectors in the previous rounds, i.e.,~$\vec{r}_{1},\dots,\vec{r}_{t-1}$.
\end{definition}
In particular, note that reward-based algorithms are oblivious to the optimizer's utility function $u_O$ and actions $i_t$ (as long as they result in the same reward vectors $\vec{r}_t$). Next, we formally define the \emph{exploitability} of reward-based algorithms with high polytope swap regret:
\begin{definition}[Exploitability] \label{def:exploitability} 
Suppose at each round $t\in[T]$ of a repeated Bayesian game $\G(M,N,C,\D,u_O,u_L)$ (where $u_O,u_L$ have range $[-1,1]$), the optimizer plays action $i_t\in[M]$, and a reward-based algorithm $\A$ uses mixed strategy $\beta_t$, and at the end of repeated game $\G$, $\A$ has polytope swap regret $\Reg_{\poly}=\Omega(T)$.
We say that $\A$ is exploitable if there exist a Bayesian game $\G'(M',N,C,\D,u_O',u_L')$ (where $u_O',u_L'$ have range $[-1,1]$) and a sequence of the optimizer's actions $i_1',\dots,i_T'\in[M']$ such that
\begin{itemize}
    \item[(i)] $u_L'(i_t',j,c)=u_L(i_t,j,c)$ for all $t\in[T]$, $j\in[N]$ and $c\in[C]$ (which implies that $\A$ also uses mixed strategies $\beta_1,\dots,\beta_T$ when facing the optimizer's actions $i_1',\dots,i_T'$ in repeated game $\G'$, because it faces the same sequence of reward vectors),
    \item[(ii)] and moreover, the optimizer's expected utility by playing actions $i_1',\dots,i_T'$ against algorithm $\A$ in repeated game $\G'$, is $V(\G')\cdot T+\Omega(T)$, where $V(\G')$ is the Stackelberg utility of $\G'$.
\end{itemize}
\end{definition}
The question of~\citet[Question 1]{mansour2022strategizing} can now be formally stated as follows: Consider the Bayesian games where the learner's and the optimizer's utilities are in the range $[-1,1]$. Is it true that if a reward-based algorithm has $\Omega(T)$ polytope swap regret in some repeated Bayesian game, then it must be exploitable? We answer this question affirmatively under a reasonable and necessary condition (indeed, in Section~\ref{section:negligible_p_c}, we will show that without this condition, the answer is negative):
\begin{definition}[No-negligible-context] \label{definition:p_c_is_not_negligible} 
Given a Bayesian game $\G(M,N,C,\D,u_O,u_L)$ where the probability of each context $c\in[C]$ is $p_c$, we say $\G$ has no negligible context if $p_c=\Omega(1)$ for all $c\in[C]$. In particular, this implies that (i) there is no zero-probability context, and (ii) $C=O(1)$.
\end{definition}
\begin{theorem}\label{thm:exploitable-poly-swap-regret}
Consider Bayesian games where the learner's and the optimizer's utilities are in the range $[-1,1]$. If a reward-based learning algorithm $\A$ has $\Omega(T)$ polytope swap regret in some repeated Bayesian game that has no negligible context, then $\A$ is exploitable.
\end{theorem}
\begin{proof}
Suppose at each round $t\in[T]$ of a repeated Bayesian game $\G(M,N,C,\D,u_O,u_L)$ ($u_O, u_L$ have range $[-1,1]$), the optimizer plays action $i_t\in[M]$, and a reward-based algorithm $\A$ uses mixed strategy $\beta_t\in\Delta([N]^{[C]})$, and in the end, $\A$ has polytope swap regret $\Reg_{\poly}(\G)=\Omega(T)$. We will prove that there exists another Bayesian game $\G'(M',N,C,\D,u_O',u_L')$ ($u_O', u_L'$ have range $[-1,1]$) such that the optimizer can achieve expected utility $V(\G')\cdot T+\Omega(T)$ when playing $\G'$ repeated for $T$ rounds against $\A$. First, we construct $u_L'$ and the optimizer's oblivious strategy for playing repeated game $\G'$ against $\A$ (on the other hand, instead of explicitly constructing $u_O'$, we will show that {\em there exists} $u_O'$ which satisfy our desiderata using a min-max program).
\begin{proof}[Construction of $u_L'$ and the optimizer's strategy]\renewcommand{\qedsymbol}{}
In the Bayesian game $\G'$, the optimizer has $M'=T$ actions $A'=\{a_1,\dots,a_T\}$. The learner's utility $u_L'$ is defined ``similarly'' to $u_L$. That is, for each optimizer's action $a_t$ with $t\in[T]$, each learner's action $j\in[N]$ and each context $c\in[C]$, we let $u_L'(a_t,j, c):=u_L(i_t,j, c)$ (in other words, the optimizer's action $a_t$ in $\G'$ implies the same reward vector as action $i_t$ in $\G$). Finally, the optimizer's strategy when playing repeated game $\G'$ against $\A$ is simply playing action $a_t$ at each round $t\in[T]$.
\end{proof}
Before we proceed, note that by our construction above, the reward vector at each round $t\in[T]$ of repeated game $\G'$ when the optimizer plays action $a_t$, is exactly the same as the reward vector at each round $t\in[T]$ of repeated game $\G$ when the optimizer plays action $i_t$, and therefore, the reward-based algorithm $\A$ uses the same mixed strategy at each round $t\in[T]$ in repeated game $\G'$ as in repeated game $\G$, which is mixed strategy $\beta_t$. Let $\Reg_{\poly}(\G')$ denote the polytope swap regret of $\A$ when playing repeated game $\G'$ against the optimizer's sequence of actions $a_1,\dots,a_T$. By Definition~\ref{def:S-poly-regret}, we have that
\begin{align}
\Reg_{\poly}(\G')&=\min_{\forall t\in[T],\,\rho_t\in\class{\beta_t}}\max_{\pi:[N]^{[C]}\to[N]^{[C]}}\sum_{t\in[T]}\sum_{f\in[N]^{[C]}}\rho_{t,f}\cdot u_L'(a_t,\pi(f))-u_L'(a_t,\beta_t)\nonumber\\
&=\min_{\forall t\in[T],\,\rho_t\in\class{\beta_t}}\max_{\pi:[N]^{[C]}\to[N]^{[C]}}\sum_{t\in[T]}\sum_{f\in[N]^{[C]}}\rho_{t,f}\cdot u_L(i_t,\pi(f))-u_L(i_t,\beta_t) \,\,\text{(By construction)}\nonumber\\
&=\Reg_{\poly}(\G)=\Omega(T).\label{eq:poly_swap_regret_G'}
\end{align}

\subsubsection*{Formulating a max-min program to find $u_O'$}
Let $\beta_t(j,c)$ denote the marginal probability that the learner plays action $j\in[N]$ conditioned on context $c\in[C]$ when the learner uses mixed strategy $\beta_t$, and recall that $p_c$ denotes the probability of context $c$ according to the prior distribution $\D$. Now if we think of $(u_O'(a_t,j,c))_{t\in[T],j\in[N],c\in[C]}$, which we have not specified, as variables, we can represent the optimizer's expected utility of playing repeated game $\G'$ against $\A$ using the oblivious strategy in our construction, which we denote by $U_O$, as a linear function of variables $(u_O'(a_t,j,c))_{t\in[T],j\in[N],c\in[C]}$,
\begin{align}\label{eq:U_O}
U_O=\sum_{t\in[T]} u_O'(a_t,\beta_t)=\sum_{t\in[T]} \sum_{c\in[C]}\sum_{j\in[N]} p_c\cdot \beta_t(j,c)\cdot u_O'(a_t,j,c).
\end{align}

On the other hand, the optimizer's expected utility, achieved using any mixed strategy $\alpha\in\Delta(A')$ against any best response $f\in[N]^{[C]}$ of the learner to $\alpha$, can also be written as a linear function of variables $(u_O'(a_t,j,c))_{t\in[T],j\in[N],c\in[C]}$ (let $\alpha_{a_t}$ denote the probability of playing action $a_t$ according to mixed strategy $\alpha$),
\begin{align}\label{eq:u_O'_alpha_f}
u_O'(\alpha, f)=\sum_{t\in [T]}\sum_{c\in[C]} p_c\cdot\alpha_{a_t}\cdot u_O'(a_t, f(c),c).
\end{align}
Note that the Stackelberg utility $V(\G')$, which we eventually want to compare $\frac{U_O}{T}$ with, is by definition the maximum of $u_O'(\alpha, f)$
over all possible pairs $(\alpha,f)$ where $f$ is a best response to $\alpha$. To simplify the analysis, instead of all possible pairs, we consider a finite subset $\J$ of such pairs, which will give an approximation of $V(\G')$. Specifically, we first partition the probability simplex $\Delta(A')$ into $(T^2+1)^T$ parts as follows: For each $x\in\{0,\frac{1}{T^2},\dots,1\}^T$, let $\Delta_x:=\{\alpha\in\Delta(A')\mid \forall t\in[T],\,\alpha_{a_t}\in[x_t,x_t+\frac{1}{T^2})\}$, and hence $\Delta(A')=\dot{\bigcup}_{x\in\{0,1,\dots,T^2\}^T}\Delta_x$ (some $\Delta_x$'s are empty, but they are irrelevant). Then, $\J$ is constructed as follows: For each $x\in\{0,\frac{1}{T^2},\dots,1\}^T$ and each $f\in[N]^{[C]}$, if there exists $\alpha\in \Delta_x$ such that $f$ is a best response to mixed strategy $\alpha$ (if there are multiple such $\alpha$, we choose an arbitrary one), we add the pair $(\alpha,f)$ to $\J$. In the following claim, we show that $\J$ indeed gives a good approximation of $V(\G')$, given any $u_O'$ with range $[-1,1]$.
\begin{claim}\label{claim:J_good_approximation}
For any $u_O'$ with range $[-1,1]$, it holds that $\max_{(\alpha,f)\in\J} u_O'(\alpha, f)+\frac{1}{T}\ge V(\G')$.
\end{claim}
\begin{proof}[Proof of Claim~\ref{claim:J_good_approximation}]
Consider any $u_O'$ that has range $[-1,1]$. Let $(\alpha^*, f^*)$ be a pair of mixed strategy and corresponding best response, which achieve the Stackelberg utility $V(\G')$. Let $x\in\{0,\frac{1}{T^2},\dots,1\}^T$ be such that $\alpha^*\in\Delta_x$. Because $\alpha^*\in\Delta_x$, and $f^*$ is a best response to $\alpha^*$, by our construction of $\J$, there should exist $(\alpha, f^*)\in\J$ such that $\alpha\in\Delta_x$. We derive that
\begin{align*}
    u_O'(\alpha^*, f^*)-u_O'(\alpha, f^*)&=\sum_{t\in [T]}\sum_{c\in[C]} p_c\cdot\alpha^*_{a_t}\cdot u_O'(a_t, f^*(c),c)-\sum_{t\in [T]}\sum_{c\in[C]} p_c\cdot\alpha_{a_t}\cdot u_O'(a_t, f^*(c),c)\\
    &=\sum_{t\in [T]}\sum_{c\in[C]} p_c\cdot(\alpha^*_{a_t}-\alpha_{a_t})\cdot u_O'(a_t, f^*(c),c)\\
    &\le\sum_{t\in [T]}\sum_{c\in[C]} p_c\cdot|\alpha^*_{a_t}-\alpha_{a_t}|\cdot |u_O'(a_t, f^*(c),c)|\\
    &\le \sum_{t\in [T]}\sum_{c\in[C]} p_c\cdot\frac{1}{T^2}\cdot 1=\frac{1}{T},
\end{align*}
where the last inequality follows by $\alpha,\alpha^*\in\Delta_x$ and $u_O'(a_t, f^*(c),c)\in[-1,1]$.
\end{proof}
We are ready to formulate a max-min program that finds $u_O'$ to prove the theorem. Specifically, to prove the theorem, we want to find $(u_O'(a_t,j,c))_{t\in[T],j\in[N],c\in[C]}$ that maximize $U_O-V(\G')\cdot T$ (note that both $U_O$ and $V(\G')$ are functions of variables $(u_O'(a_t,j,c))_{t\in[T],j\in[N],c\in[C]}$) and show that the maximum is $\Omega(T)$. By Claim~\ref{claim:J_good_approximation}, $U_O-V(\G')\cdot T\ge U_O-(\max_{(\alpha,f)\in\J} u_O'(\alpha, f)\cdot T+1)=\min_{(\alpha,f)\in\J} U_O-u_O'(\alpha, f)\cdot T-1$. Therefore, to prove the theorem, it suffices to show that
\[
\max_{\substack{\forall t\in[T],j\in[N],c\in[C] \\ u_O'(a_t,j,c)\in[-1,1]}}\min_{(\alpha,f)\in\J} U_O-u_O'(\alpha, f)\cdot T = \Omega(T).
\]
Moreover, it holds straightforwardly that
\[
\max_{\substack{\forall t\in[T],j\in[N],c\in[C] \\ u_O'(a_t,j,c)\in[-1,1]}}\min_{(\alpha,f)\in\J} U_O-u_O'(\alpha, f)\cdot T=\max_{\substack{\forall t\in[T],j\in[N],c\in[C] \\ u_O'(a_t,j,c)\in[-1,1]}} \min_{\D(\J)\in\Delta(\J)}\E_{(\alpha,f)\sim \D(\J)}[U_O-u_O'(\alpha, f)\cdot T].
\]
We notice that in the max-min program on R.H.S., both $U_O$ and $u_O'(\alpha, f)$ are linear functions of variables $(u_O'(a_t,j,c))_{t\in[T],j\in[N],c\in[C]}$ as we showed in Eq.~\eqref{eq:U_O} and~\eqref{eq:u_O'_alpha_f}, and clearly $\E_{(\alpha,f)\sim \D(\J)}[U_O-u_O'(\alpha, f)\cdot T]$ is a linear function of the probabilities that specify $\D(\J)$. The constraints also clearly correspond to compact convex subsets of Euclidean space. Therefore, by Minimax theorem (e.g.,~\citet{sion1958general}), we have that
\begin{align}
&\max_{\substack{\forall t\in[T],j\in[N],c\in[C] \\ u_O'(a_t,j,c)\in[-1,1]}}\min_{\D(\J)\in\Delta(\J)} \E_{(\alpha,f)\sim \D(\J)}[U_O-u_O'(\alpha, f)\cdot T]\nonumber\\
&=\min_{\D(\J)\in\Delta(\J)}\max_{\substack{\forall t\in[T],j\in[N],c\in[C] \\ u_O'(a_t,j,c)\in[-1,1]}} \E_{(\alpha,f)\sim \D(\J)}[U_O-u_O'(\alpha, f)\cdot T].\label{eq:min_max}
\end{align}
Hence, it suffices to prove the optimal value of the min-max program in Eq.~\eqref{eq:min_max} is $\Omega(T)$.

\subsubsection*{Lower bounding the optimal value of the min-max program}
Given any distribution $\D(\J)$ that is feasible in the min-max program in Eq.~\eqref{eq:min_max}, let $\Pr[\alpha,f]$ denote the probability of each $(\alpha,f)\in\J$ according to $\D(\J)$. Then, for $t\in[T]$, let $\eta_t:=\sum_{(\alpha,f)\in\J}\alpha_{a_t}\cdot \Pr[\alpha,f]$ (one can think of $\eta_t$ as the marginal probability of optimizer's action $a_t$ according to $\D(\J)$), and let $\hat{\beta}_t$ denote the learner's mixed strategy such that the probability of playing pure strategy $f\in[N]^{[C]}$ is $\hat{\beta}_{t,f}=\frac{\sum_{(\alpha,\hat{f})\in\J\textrm{ s.t. }\hat{f}=f}\alpha_{a_t}\cdot \Pr[\alpha,f]}{\eta_t}$ if $\eta_t>0$ (in this case, one can think of $\hat{\beta}_t$ as the learner's mixed strategy according to distribution $\D(\J)$ conditioned on the optimizer's action $a_t$), and $\hat{\beta}_{t,f}=\beta_{t,f}$ if $\eta_t=0$ (this case is only for the convenience of the analysis). Similar to $\beta_t(j,c)$, we also let $\hat{\beta}_t(j,c)$ denote the marginal probability that the learner plays action $j\in[N]$ conditioned on context $c\in[C]$ according to mixed strategy $\hat{\beta}$. Now we lower bound the optimal value of the min-max program in Eq.~\eqref{eq:min_max} under certain assumptions on $\D(\J)$.
\begin{claim}\label{claim:large_L1_distance_cases}
If $\D(\J)$ satisfies either of the following conditions:
\begin{enumerate}
    \item[(i)] $\sum_{t\in[T]} |1-\eta_t\cdot T|=\Omega(T)$,
    \item[(ii)] $\sum_{t\in[T]} |1-\eta_t\cdot T|=o(T)$ and $\sum_{t\in[T]}\sum_{j\in[N]}\sum_{c\in[C]} p_c\cdot |\beta_t(j,c)-\hat{\beta}_t(j,c)|=\Omega(T)$,
\end{enumerate}
then $\max\limits_{\substack{\forall t\in[T],j\in[N],c\in[C] \\ u_O'(a_t,j,c)\in[-1,1]}} \E\limits_{(\alpha,f)\sim \D(\J)}[U_O-u_O'(\alpha, f)\cdot T]=\Omega(T).$
\end{claim}
\begin{proof}[Proof of Claim~\ref{claim:large_L1_distance_cases}]
First, we derive $\E_{(\alpha,f)\sim \D(\J)}[u_O'(\alpha, f)]$ as follows (intuitively, this is probably straightforward from our intuitive explanation of $\eta_t$ and $\hat{\beta}_t$)
\begin{align}
    \E\limits_{(\alpha,f)\sim \D(\J)}[u_O'(\alpha, f)]&=\E\limits_{(\alpha,f)\sim \D(\J)}[\sum_{t\in[T]}\alpha_{a_t}\cdot u_O'(a_t, f)]\nonumber\\ &=\sum_{(\alpha,f)\in\J}\Pr[\alpha,f]\cdot\sum_{t\in[T]}\alpha_{a_t}\cdot u_O'(a_t, f)\nonumber\\
    &=\sum_{t\in[T]\textrm{ s.t. }\eta_t>0}\eta_t\cdot\frac{\sum_{(\alpha,f)\in\J}\alpha_{a_t}\cdot\Pr[\alpha,f]\cdot u_O'(a_t, f)}{\eta_t}\nonumber\\
    &=\sum_{t\in[T]\textrm{ s.t. }\eta_t>0}\eta_t\cdot\frac{\sum_{\hat{f}\in[N]^{[C]}}\sum_{(\alpha,f)\in\J\textrm{ s.t. }f=\hat{f}}\alpha_{a_t}\cdot\Pr[\alpha,\hat{f}]\cdot u_O'(a_t, \hat{f})}{\eta_t}\nonumber\\
    &=\sum_{t\in[T]\textrm{ s.t. }\eta_t>0}\eta_t\cdot\sum_{\hat{f}\in[N]^{[C]}}\hat{\beta}_{t,\hat{f}}\cdot u_O'(a_t, \hat{f})\nonumber\\
    &=\sum_{t\in[T]\textrm{ s.t. }\eta_t>0}\eta_t\cdot u_O'(a_t,\hat{\beta}_t)\nonumber\\
    &=\sum_{t\in[T]}\eta_t\cdot u_O'(a_t,\hat{\beta}_t)=\sum_{t\in[T]}\eta_t\cdot\sum_{j\in[N]}\sum_{c\in[C]}p_c\cdot\hat{\beta}_t(j,c)\cdot u_O'(a_t,j,c)\label{eq:expected_utility_D(I')}.
\end{align}
By Eq.~\eqref{eq:U_O} and~\eqref{eq:expected_utility_D(I')}, we have that
\[
\E\limits_{(\alpha,f)\sim \D(\J)}[U_O-u_O'(\alpha, f)\cdot T]=\sum_{t\in [T]}\sum_{j\in[N]}\sum_{c\in[C]}p_c\cdot(\beta_t(j,c)-\eta_t\cdot T\cdot\hat{\beta}_t(j,c))\cdot u_O'(a_t,j,c).
\]
Under condition (i), we can choose $u_O'(a_t,j,c)$ to be $1$ if $\eta_t\cdot T\le 1$ and $-1$ if otherwise, then we get
\begin{align*}
\E\limits_{(\alpha,f)\sim \D(\J)}[U_O-u_O'(\alpha, f)\cdot T]&=\sum_{\substack{t\in [T] \\\textrm{s.t. } \eta_t\cdot T\le 1}}\sum_{c\in[C]}\sum_{j\in[N]}p_c\cdot(\beta_t(j,c)-\eta_t\cdot T\cdot\hat{\beta}_t(j,c))\\
&\quad-\sum_{\substack{t\in [T] \\\textrm{s.t. } \eta_t\cdot T> 1}}\sum_{c\in[C]}\sum_{j\in[N]}p_c\cdot(\beta_t(j,c)-\eta_t\cdot T\cdot\hat{\beta}_t(j,c))\\
&=\sum_{\substack{t\in [T] \\\textrm{s.t. } \eta_t\cdot T\le 1}}\sum_{c\in[C]}p_c\cdot(1-\eta_t\cdot T)-\sum_{\substack{t\in [T] \\\textrm{s.t. } \eta_t\cdot T> 1}}\sum_{c\in[C]}p_c\cdot(1-\eta_t\cdot T)\\
&\qquad\qquad\qquad\qquad\qquad\qquad\qquad\text{(By $\sum_{j\in[N]} \beta_t(j,c)=\sum_{j\in[N]} \hat{\beta}_t(j,c)=1$)}\\
&=\bigg(\sum_{\substack{t\in [T] \\\textrm{s.t. } \eta_t\cdot T\le 1}}1-\eta_t\cdot T\bigg)-\bigg(\sum_{\substack{t\in [T] \\\textrm{s.t. } \eta_t\cdot T> 1}}1-\eta_t\cdot T\bigg)\\
&=\sum_{t\in[T]}|1-\eta_t\cdot T|=\Omega(T).
\end{align*}
Under condition (ii), we can choose $u_O'(a_t,j,c)$ to be $1$ if $\beta_t(j,c)\ge\eta_t\cdot T\cdot\hat{\beta}_t(j,c)$ and $-1$ if otherwise, and we get
\begin{align*}
&\E\limits_{(\alpha,f)\sim \D(\J)}[U_O-u_O'(\alpha, f)\cdot T]\\
=&\sum_{t\in [T]}\sum_{j\in[N]}\sum_{c\in[C]}p_c\cdot|\beta_t(j,c)-\eta_t\cdot T\cdot \hat{\beta}_t(j,c)|\\
\ge&\sum_{t\in [T]}\sum_{j\in[N]}\sum_{c\in[C]}p_c\cdot(|\beta_t(j,c)-\hat{\beta}_t(j,c)| - |\eta_t\cdot T\cdot \hat{\beta}_t(j,c)-\hat{\beta}_t(j,c)|)\qquad\qquad\text{(Triangle inequality)}\\
=&\sum_{t\in [T]}\sum_{j\in[N]}\sum_{c\in[C]}p_c\cdot|\beta_t(j,c)-\hat{\beta}_t(j,c)|-\sum_{t\in [T]}|1-\eta_t\cdot T|\cdot\sum_{c\in[C]}\sum_{j\in[N]}p_c\cdot \hat{\beta}_t(j,c)\\
=&\sum_{t\in [T]}\sum_{j\in[N]}\sum_{c\in[C]}p_c\cdot|\beta_t(j,c)-\hat{\beta}_t(j,c)|-\sum_{t\in [T]}|1-\eta_t\cdot T|=\Omega(T).
\end{align*}
\end{proof}
By Claim~\ref{claim:large_L1_distance_cases}, it remains to handle the case when $\D(\J)$ satisfies both $\sum_{t\in T}|1-\eta_t\cdot T|=o(T)$ and $\sum_{t\in[T]}\sum_{j\in[N]}\sum_{c\in[C]} p_c\cdot |\beta_t(j,c)-\hat{\beta}_t(j,c)|=o(T)$. We will show that this case contradicts the assumption that $\A$ has $\Omega(T)$ polytope swap regret in repeated game $\G'$ (i.e., Eq.~\eqref{eq:poly_swap_regret_G'}).
\subsubsection*{Deriving a contradiction with $\Omega(T)$ polytope swap regret}
We define $\widehat{\Reg_{\poly}}(\G')$ and $\widehat{\Reg_{\poly}}'(\G')$ as follows (essentially, we get $\widehat{\Reg_{\poly}}(\G')$ by replacing all the $\beta_t$ in $\Reg_{\poly}(\G')$ with $\hat{\beta}_t$ and then get $\widehat{\Reg_{\poly}}'(\G')$ by re-weighting $\widehat{\Reg_{\poly}}(\G')$),
\begin{align}
&\widehat{\Reg_{\poly}}(\G'):=\min_{\forall t\in[T],\,\rho_t\in\class{\hat{\beta}_t}}\max_{\pi:[N]^{[C]}\to[N]^{[C]}}\sum_{t\in[T]}\sum_{f\in[N]^{[C]}}\rho_{t,f}\cdot u_L'(a_t,\pi(f))-u_L'(a_t,\hat{\beta}_t),\label{eq:hat_reg}\\
&\widehat{\Reg_{\poly}}'(\G'):=\min_{\forall t\in[T],\,\rho_t\in\class{\hat{\beta}_t}}\max_{\pi:[N]^{[C]}\to[N]^{[C]}}\sum_{t\in[T]}\eta_t\cdot T\cdot\left(\sum_{f\in[N]^{[C]}}\rho_{t,f}\cdot u_L'(a_t,\pi(f))-u_L'(a_t,\hat{\beta}_t)\right).\label{eq:hat_reg'}
\end{align}
Our plan to show that, $\sum_{t\in [T]}|1-\eta_t\cdot T|=o(T)$ and $\sum_{t\in[T]}\sum_{j\in[N]}\sum_{c\in[C]} p_c\cdot |\beta_t(j,c)-\hat{\beta}_t(j,c)|=o(T)$ together contradict with $\Reg_{\poly}(\G')=\Omega(T)$, is the following: First, we show that if $\sum_{t\in [T]}|1-\eta_t\cdot T|=o(T)$, then $\widehat{\Reg_{\poly}}(\G')\le\widehat{\Reg_{\poly}}'(\G')+o(T)$ (which is rather straightforward). Then, we prove that if $\sum_{t\in[T]}\sum_{j\in[N]}\sum_{c\in[C]} p_c\cdot |\beta_t(j,c)-\hat{\beta}_t(j,c)|=o(T)$, then $\Reg_{\poly}(\G')\le\widehat{\Reg_{\poly}}(\G')+o(T)$ (which is the most technical part of the proof, but intuitively this simply means that the polytope swap regret does not change much after ``small perturbations'' to the learner's mixed strategies $\beta_t$'s). Finally, we derive a contradiction by showing $\widehat{\Reg_{\poly}}'(\G')=0$ (intuitively, this is because $\widehat{\Reg_{\poly}}'(\G')$ is closely related to $\D(\J)$, and $\D(\J)$ is a distribution over pairs of the optimizer's mixed strategy and the learner's corresponding best response, which provides certain way to decompose $\hat{\beta}_t$'s such that there is no polytope swap regret in the sense of $\widehat{\Reg_{\poly}}'(\G')$).

Now we implement our plan through the following three claims.

\begin{claim}\label{claim:hat_reg_near_hat_reg'}
If $\D(\J)$ satisfies that $\sum_{t\in [T]}|1-\eta_t\cdot T|=o(T)$, then $\widehat{\Reg_{\poly}}(\G')\le\widehat{\Reg_{\poly}}'(\G')+o(T)$.
\end{claim}
\begin{proof}[Proof of Claim~\ref{claim:hat_reg_near_hat_reg'}]
First, we define functions 
\begin{align*}
&g((\rho_t)_{t\in[T]},\pi):=\sum_{t\in[T]}\sum_{f\in[N]^{[C]}}\rho_{t,f}\cdot u_L'(a_t,\pi(f))-u_L'(a_t,\hat{\beta}_t),\\
&g'((\rho_t)_{t\in[T]},\pi):=\sum_{t\in[T]}\eta_t\cdot T\cdot\left(\sum_{f\in[N]^{[C]}}\rho_{t,f}\cdot u_L'(a_t,\pi(f))-u_L'(a_t,\hat{\beta}_t)\right),
\end{align*}
and we notice that
\begin{align*}
    g((\rho_t)_{t\in[T]},\pi)-g'((\rho_t)_{t\in[T]},\pi)&=\sum_{t\in[T]}(1-\eta_t\cdot T)\cdot\left(\sum_{f\in[N]^{[C]}}\rho_{t,f}\cdot u_L'(a_t,\pi(f))-u_L'(a_t,\hat{\beta}_t)\right)\nonumber\\
    &\le\sum_{t\in[T]}|1-\eta_t\cdot T|\cdot\left|\sum_{f\in[N]^{[C]}}\rho_{t,f}\cdot u_L'(a_t,\pi(f))-u_L'(a_t,\hat{\beta}_t)\right|\nonumber\\
    &\le\sum_{t\in[T]}|1-\eta_t\cdot T|\cdot\left(\sum_{f\in[N]^{[C]}}\rho_{t,f}\cdot 1+1\right) \quad\quad\text{($u_L'$ has range $[-1,1]$)}\nonumber\\
    &=2\cdot\sum_{t\in[T]}|1-\eta_t\cdot T|=o(T).
\end{align*}
Therefore, we have that $\max_{\pi:[N]^{[C]}\to[N]^{[C]}} g((\rho_t)_{t\in[T]},\pi)\le\max_{\pi:[N]^{[C]}\to[N]^{[C]}} g'((\rho_t)_{t\in[T]},\pi)+o(T)$ for any $(\rho_t)_{t\in[T]}$, and
it follows that $$\min_{\forall t\in[T],\,\rho_t\in\class{\hat{\beta}_t}}\max_{\pi:[N]^{[C]}\to[N]^{[C]}} g((\rho_t)_{t\in[T]},\pi)\le \min_{\forall t\in[T],\,\rho_t\in\class{\hat{\beta}_t}}\max_{\pi:[N]^{[C]}\to[N]^{[C]}} g'((\rho_t)_{t\in[T]},\pi)+o(T).$$
The proof finishes because the L.H.S.~is exactly $\widehat{\Reg_{\poly}}(\G')$, and the R.H.S.~is exactly $\widehat{\Reg_{\poly}}'(\G')$.
\end{proof}

\begin{claim}\label{claim:reg_near_hat_reg}
If $\D(\J)$ satisfies that $\sum_{t\in[T]}\sum_{j\in[N]}\sum_{c\in[C]} p_c\cdot |\beta_t(j,c)-\hat{\beta}_t(j,c)|=o(T)$, then it holds that $\Reg_{\poly}(\G')\le\widehat{\Reg_{\poly}}(\G')+o(T)$.
\end{claim}
\begin{proof}[Proof of Claim~\ref{claim:reg_near_hat_reg}]
Let $(\hat{\rho}_t)_{t\in[T]}$ and $\hat{\pi}$ be the optimal solution to the min-max program that defines $\widehat{\Reg_{\poly}}(\G')$ (see Eq.~\eqref{eq:hat_reg}). We will construct $\rho_t\in\class{\beta_t}$ for all $t\in[T]$ such that
\begin{align}
&\max_{\pi:[N]^{[C]}\to[N]^{[C]}}\sum_{t\in[T]}\sum_{f\in[N]^{[C]}}\rho_{t,f}\cdot u_L'(a_t,\pi(f))-u_L'(a_t,\beta_t)\nonumber\\
\le&\sum_{t\in[T]}\sum_{f\in[N]^{[C]}}\hat{\rho}_{t,f}\cdot u_L'(a_t,\hat{\pi}(f))-u_L'(a_t,\hat{\beta}_t)+o(T).\label{eq:reg_near_hat_reg_goal}
\end{align}
Notice that this will prove the claim because the L.H.S.~trivially upper bounds $\Reg_{\poly}(\G')$ by definition (see Eq.~\eqref{eq:poly_swap_regret_G'}), and the R.H.S.~is obviously $\widehat{\Reg_{\poly}}(\G')+o(T)$.

We start by constructing $\rho_t\in\class{\beta_t}$ for each $t\in[T]$ with bounded statistical distance from $\hat{\rho}_t$ in the following lemma.
\begin{lemma}\label{lemma:construct_rho}
For any $\hat{\rho}_t\in\class{\hat{\beta}_t}$, there exists $\rho_t\in\class{\beta_t}$ such that 
\begin{equation}\label{eq:construct_rho}
\sum_{f\in[N]^{[C]}}|\rho_{t,f}-\hat{\rho}_{t,f}|\le2\cdot\sum_{j\in[N]}\sum_{c\in[C]}|\beta_t(j,c)-\hat{\beta}_t(j,c)|.
\end{equation}
\end{lemma}
\begin{proof}[Proof of Lemma~\ref{lemma:construct_rho}]
We construct $\rho_{t}$ in two steps: First, we construct $\rho_{t,f}^{(1)}$ for all $f\in[N]^{[C]}$ that satisfy
\begin{align*}
&\textrm{Property (i): }\rho_{t,f}^{(1)}\le\hat{\rho}_{t,f},\\
&\textrm{Property (ii): }1-\sum_{f\in[N]^{[C]}}\rho_{t,f}^{(1)}\le\sum_{j\in[N]}\sum_{c\in[C]}|\beta_t(j,c)-\hat{\beta}_t(j,c)|.
\end{align*}
Then, we will construct $\rho_{t,f}^{(2)}$ for all $f\in[N]^{[C]}$ such that $\rho_{t,f}=\rho_{t,f}^{(1)}+\rho_{t,f}^{(2)}$. Finally, we will prove that the constructed $\rho_t$ satisfies Ineq.~\eqref{eq:construct_rho}.

\subsubsection*{Constructing $\rho_{t,f}^{(1)}$}
We order all the $f\in[N]^{[C]}$ arbitrarily, which we denote by $f_1,\dots,f_{N^C}$, and then we use Procedure~\ref{procedure:construct_rho} to construct $\rho_{t,f}^{(1)}$ for all $f\in[N]^{[C]}$. The main idea of Procedure~\ref{procedure:construct_rho} is simply iterating over $f_1,\dots,f_{N^C}$ and letting each $\rho_{t,f_i}^{(1)}$ be the maximum value that is ``capped'' by $\hat{\rho}_{t,f_i}$ and $\beta_t$.
\begin{algorithm}[ht]
\SetAlgoLined
\SetAlgorithmName{Procedure}~~
Initialize $\rho_{t,f}^{(1)}=0$ for all $f\in[N]^{[C]}$, $\rho_t^{(1)}(j,c)=0$ for all $j\in[N]$ and $c\in[C]$, and $E=\{f\in[N]^{[C]}\mid \beta_{t}(f(c),c)>0\}$\;
\For{$i\in[N^C]$}{
    \If{$f_i\in E$}{
        $\rho_{t,f_i}^{(1)}\gets \min\{\hat{\rho}_{t,f_i},\,\min_{c\in[C]}\beta_{t}(f_i(c),c)-\rho_t^{(1)}(f_i(c),c)\}$\;\label{line:construct_rho}
        \For{$c\in[C]$}{
            $\rho_t^{(1)}(f_i(c),c)\gets\rho_t^{(1)}(f_i(c),c)+\rho_{t,f_i}^{(1)}$\;\label{line:track_marginal}
            \If{$\beta_{t}(f_i(c),c)=\rho_t^{(1)}(f_i(c),c)$}{
                $E=E\setminus\{f\in[N]^{[C]}\mid f(c)=f_i(c)\}$\;\label{line:track_active_f}
            }
        }
    }
}
\caption{Constructing $\rho_{t,f}^{(1)}$ for all $f\in[N]^{[C]}$}
\label{procedure:construct_rho}
\end{algorithm}
Specifically, throughout Procedure~\ref{procedure:construct_rho}, each variable $\rho_t^{(1)}(j,c)$ keeps track of the cumulative probability of action $j$ conditioned on context $c$ according to the $\rho_{t,f}^{(1)}$'s, i.e., we keep the invariant 
\[
\textrm{Invariant (1): }\rho_t^{(1)}(j,c)=\sum_{f\in[N]^{[C]}\textrm{ s.t. }f(c)=j}\rho_{t,f}^{(1)},
\]
and the subset $E$ maintains all the $f\in[N]^{[C]}$ that satisfy $\rho_t^{(1)}(f(c),c)<\beta_{t}(f(c),c)$ for all $c\in[C]$, i.e., $E$ consists of those $f\in[N]^{[C]}$ for which we can strictly increase $\rho_{t,f}^{(1)}$ while keeping the invariant
\[
\textrm{Invariant (2): }\rho_t^{(1)}(f(c),c)\le\beta_{t}(f(c),c) \textrm{ for all } c\in[C].
\]
The outer for loop iterates over $f\in[N]^{[C]}$ and checks if $f\in E$, and if so, Line~\ref{line:construct_rho} increases $\rho_{t,f}^{(1)}$ to the maximum possible value that does not exceed $\hat{\rho}_{t,f}$ or violate $\rho_t^{(1)}(f(c),c)\le\beta_{t}(f(c),c)$ for any $c\in[C]$. Then, Line~\ref{line:track_marginal} and Line~\ref{line:track_active_f} update $\rho_t^{(1)}(j,c)$'s and $E$ correspondingly to maintain the invariants mentioned above.

Note that Property (i), $\rho_{t,f}^{(1)}\le\hat{\rho}_{t,f}$ for all $f\in[N]^{[C]}$, is a trivial consequence of Line~\ref{line:construct_rho}. Now we prove $\sum_{f\in[N]^{[C]}}\hat{\rho}_{t,f}-\rho_{t,f}^{(1)}\le\sum_{j\in[N]}\sum_{c\in[C]}|\beta_t(j,c)-\hat{\beta}_t(j,c)|$, which is equivalent to Property (ii) since $\sum_{f\in[N]^{[C]}}\hat{\rho}_{t,f}=1$. To this end, we label some of $f\in[N]^{[C]}$ as follows: For each $f\in[N]^{[C]}$, consider the very first time $f$ is removed from $E$ in Procedure~\ref{procedure:construct_rho}, which can only happen at Line~\ref{line:track_active_f} when the condition $\beta_t(f(c),c)=\rho_t^{(1)}(f(c),c)$ is true for some $c\in[C]$, in which case we label $f$ with $(f(c),c)$. If $f$ is never removed from $E$, then we keep $f$ unlabeled. Notice that if $f$ (say $f=f_i$ for $i\in[N^C]$) is never removed from $E$, then in particular, $f_i$ is not removed during the $i$-th iteration of the outer for loop of Procedure~\ref{procedure:construct_rho}, which implies $\rho_{t,f_i}^{(1)}=\hat{\rho}_{t,f_i}$ (because otherwise by Line~\ref{line:construct_rho}, $\rho_{t,f_i}^{(1)}=\beta_t(f_i(c),c)-\rho_t^{(1)}(f_i(c),c)$ for some $c\in[C]$, and thus $\beta_t(f_i(c),c)=\rho_t^{(1)}(f_i(c),c)$ after the update at Line~\ref{line:track_marginal}, and then $f_i$ should be removed from $E$ in that iteration). Moreover, for each $(j,c)\in[N]\times[C]$ such that there exists some $g\in[N]^{[C]}$ that is labeled by $(j,c)$ (which implies $\beta_t(j,c)=\rho_t^{(1)}(j,c)$ because this is the only way $g$ gets labeled by $(j,c)$ in our labeling process), we have that
\begin{align*}
&\sum_{\substack{f\in[N]^{[C]} \\ \textrm{labeled by } (j,c)}}\hat{\rho}_{t,f}-\rho^{(1)}_{t,f}\\
&\le\sum_{\substack{f\in[N]^{[C]} \\ \textrm{s.t. } f(c)=j}}\hat{\rho}_{t,f}-\rho^{(1)}_{t,f}&&\text{(By Property (i), and $f$ is labeled by $(j,c)$ only if $f(c)=j$)}\\
&=\hat{\beta}_t(j,c)-\rho_t^{(1)}(j,c)&&\text{(By $\hat{\beta}_t(j,c)=\sum_{\substack{f\in[N]^{[C]} \\ \textrm{s.t. } f(c)=j}}\hat{\rho}_{t,f}$, and $\rho_t^{(1)}(j,c)=\sum_{\substack{f\in[N]^{[C]} \\ \textrm{s.t.} f(c)=j}}\rho^{(1)}_{t,f}$)}\\
&=\hat{\beta}_t(j,c)-\beta_t(j,c)&&\text{(By $\rho_t^{(1)}(j,c)=\beta_t(j,c)$)}.
\end{align*}
Therefore, we have that
\begin{align*}
    \sum_{f\in[N]^{[C]}}\hat{\rho}_{t,f}-\rho^{(1)}_{t,f}&=\sum_{\substack{f\in[N]^{[C]} \\ \textrm{unlabeled}}}\hat{\rho}_{t,f}-\rho^{(1)}_{t,f} + \sum_{\substack{j\in[N],c\in[C] \\ \textrm{s.t. } \exists g\in[N]^{[C]} \\ \textrm{labeled by } (j,c)}}\sum_{\substack{f\in[N]^{[C]} \\ \textrm{labeled by } (j,c)}}\hat{\rho}_{t,f}-\rho^{(1)}_{t,f}\\
    &=0 + \sum_{\substack{j\in[N],c\in[C] \\ \textrm{s.t. } \exists g\in[N]^{[C]} \\ \textrm{labeled by } (j,c)}}\hat{\beta}_t(j,c)-\beta_t(j,c)\le\sum_{j\in[N]}\sum_{c\in[C]}|\hat{\beta}_t(j,c)-\beta_t(j,c)|.
\end{align*}

\subsubsection*{Constructing $\rho_{t,f}^{(2)}$}
Next, we proceed to construct $\rho_{t,f}^{(2)}$ for all $f\in[N]^{[C]}$ such that $\rho_{t}$ defined by $\rho_{t,f}=\rho_{t,f}^{(1)}+\rho_{t,f}^{(2)}$ is indeed equivalent to $\beta_t$, i.e., $\rho_{t}(j,c)=\beta_{t}(j,c)$ for all $j\in[N]$ and $c\in[C]$ (where $\rho_{t}(j,c):=\sum_{f\in[N]^{[C]}\textrm{ s.t. }f(c)=j}\rho_{t,f}$). First, note that if $\sum_{f\in[N]^{[C]}}\rho_{t,f}^{(1)}=1$, then for all $c\in[C]$
\begin{align*}
    \sum_{j\in[N]}\rho_t^{(1)}(j,c)=\sum_{j\in[N]}\sum_{f\in[N]^{[C]}\textrm{ s.t. }f(c)=j}\rho_{t,f}^{(1)}
    =\sum_{f\in[N]^{[C]}}\rho_{t,f}^{(1)}
    =1=\sum_{j\in[N]}\beta_t(j,c),
\end{align*}
and by Invariant (2), we must have $\rho_t^{(1)}(j,c)=\beta_t(j,c)$ for all $j\in[N]$. Therefore, in this case, we can simply let $\rho_{t,f}^{(2)}=0$, and then $\rho_{t}$ defined by $\rho_{t,f}=\rho_{t,f}^{(1)}+\rho_{t,f}^{(2)}=\rho_{t,f}^{(1)}$ is equivalent to $\beta_t$ since $\rho_t(j,c)=\rho_t^{(1)}(j,c)=\beta_t(j,c)$ for all $j\in[N],c\in[C]$. If otherwise $\sum_{f\in[N]^{[C]}}\rho_{t,f}^{(1)}<1$ ($\sum_{f\in[N]^{[C]}}\rho_{t,f}^{(1)}>1$ is impossible because $\sum_{f\in[N]^{[C]}}\rho_{t,f}^{(1)}\le\sum_{f\in[N]^{[C]}}\hat{\rho}_{t,f}=1$ by Property (i)), the construction is also straightforward: We first let $\rho_{t}^{(2)}(j,c)=\beta_{t}(j,c)-\rho_{t}^{(1)}(j,c)$ for all $j\in[N],c\in[C]$ (note that $\rho_{t}^{(2)}(j,c)\ge0$ by Invariant (2)), and then we let $\rho_{t,f}^{(2)}=\frac{\prod_{c\in[C]}\rho_{t}^{(2)}(f(c),c)}{(1-\sum_{f\in[N]^{[C]}}\rho_{t,f}^{(1)})^{C-1}}$ for all $f\in[N]^{[C]}$ (note this value is well-defined and non-negative, since $\sum_{f\in[N]^{[C]}}\rho_{t,f}^{(1)}<1$, and $\rho_{t}^{(2)}(j,c)$'s are non-negative). By this construction, we have that
\begin{align*}
\sum_{f\in[N]^{[C]}\textrm{ s.t. }f(c)=j}\rho_{t,f}^{(2)}&=\rho_{t}^{(2)}(j,c)\cdot\frac{\sum_{f'\in[N]^{[C]\setminus\{c\}}}\prod_{c'\in[C]\setminus\{c\}}\rho_{t}^{(2)}(f'(c'),c')}{(1-\sum_{f\in[N]^{[C]}}\rho_{t,f}^{(1)})^{C-1}}\\
&=\rho_{t}^{(2)}(j,c)\cdot\frac{\prod_{c'\in[C]\setminus\{c\}}(\sum_{j\in[N]}\rho_{t}^{(2)}(j,c'))}{(1-\sum_{f\in[N]^{[C]}}\rho_{t,f}^{(1)})^{C-1}}\\
&=\rho_{t}^{(2)}(j,c)\cdot\frac{\prod_{c'\in[C]\setminus\{c\}}(\sum_{j\in[N]}\beta_{t}(j,c')-\rho_{t}^{(1)}(j,c'))}{(1-\sum_{f\in[N]^{[C]}}\rho_{t,f}^{(1)})^{C-1}}\\
&=\rho_{t}^{(2)}(j,c)\cdot\frac{\prod_{c'\in[C]\setminus\{c\}}(1-\sum_{j\in[N]}\rho_{t}^{(1)}(j,c'))}{(1-\sum_{f\in[N]^{[C]}}\rho_{t,f}^{(1)})^{C-1}}\\
&=\rho_{t}^{(2)}(j,c)\cdot\frac{\prod_{c'\in[C]\setminus\{c\}}(1-\sum_{f\in[N]^{[C]}}\rho_{t,f}^{(1)})}{(1-\sum_{f\in[N]^{[C]}}\rho_{t,f}^{(1)})^{C-1}}=\rho_{t}^{(2)}(j,c),
\end{align*}
and hence, in this case, $\rho_{t}$ defined by $\rho_{t,f}=\rho_{t,f}^{(1)}+\rho_{t,f}^{(2)}$ is also equivalent to $\beta_t$ since $\rho_{t}(j,c)=\rho_{t}^{(1)}(j,c)+\rho_{t}^{(2)}(j,c)=\rho_{t}^{(1)}(j,c)+(\beta_{t}(j,c)-\rho_{t}^{(1)}(j,c))=\beta_{t}(j,c)$ for all $j\in[N],c\in[C]$.

\subsubsection*{Proving Ineq.~\eqref{eq:construct_rho}}
Having constructed $\rho_t$, we are ready to show that it satisfies Ineq.~\eqref{eq:construct_rho}. To this end, for all $f\in[N]^{[C]}$, we let $\hat{\rho}^{(2)}_{t,f}:=\hat{\rho}_{t,f}-\rho^{(1)}_{t,f}$, which is non-negative by Property (i), and we derive that
\begin{align}\label{eq:bound_l1_hat_rho_vs_rho}
\sum_{f\in[N]^{[C]}}|\rho_{t,f}-\hat{\rho}_{t,f}|&=\sum_{f\in[N]^{[C]}}|\rho^{(1)}_{t,f}+\rho^{(2)}_{t,f}-\hat{\rho}_{t,f}|\nonumber\\
&=\sum_{f\in[N]^{[C]}}|\rho^{(2)}_{t,f}-\hat{\rho}^{(2)}_{t,f}|\nonumber\\
&\le\sum_{f\in[N]^{[C]}}|\rho^{(2)}_{t,f}|+|\hat{\rho}^{(2)}_{t,f}|\nonumber\\
&=\sum_{f\in[N]^{[C]}}\rho^{(2)}_{t,f}+\hat{\rho}^{(2)}_{t,f}.
\end{align}
Furthermore, we notice that 
\begin{align*}
&\sum_{f\in[N]^{[C]}}\hat{\rho}^{(2)}_{t,f}=\sum_{f\in[N]^{[C]}}\hat{\rho}_{t,f}-\rho^{(1)}_{t,f}=1-\sum_{f\in[N]^{[C]}}\rho^{(1)}_{t,f},\\
&\sum_{f\in[N]^{[C]}}\rho^{(2)}_{t,f}=\sum_{f\in[N]^{[C]}}\beta_{t,f}-\rho^{(1)}_{t,f}=1-\sum_{f\in[N]^{[C]}}\rho^{(1)}_{t,f}.
\end{align*}
and then it follows by Ineq.~\eqref{eq:bound_l1_hat_rho_vs_rho} that
$\sum_{f\in[N]^{[C]}}|\rho_{t,f}-\hat{\rho}_{t,f}|\le 2\cdot(1-\sum_{f\in[N]^{[C]}}\rho^{(1)}_{t,f})$, which is at most $2\cdot\sum_{j\in[N]}\sum_{c\in[C]}|\beta_t(j,c)-\hat{\beta}_t(j,c)|$ by Property (ii).
\end{proof}
Now we prove that the $\rho_t$'s constructed by Lemma~\ref{lemma:construct_rho} satisfy Ineq.~\eqref{eq:reg_near_hat_reg_goal}, which will finish the proof of Claim~\ref{claim:reg_near_hat_reg} as we explained. To this end, let $\pi:[N]^{[C]}\to[N]^{[C]}$ denote the optimal solution to the maximization program on L.H.S.~of Ineq.~\eqref{eq:reg_near_hat_reg_goal}. Since $(\hat{\rho}_t)_{t\in[T]}$ and $\hat{\pi}$ are optimal solutions to the min-max program of Eq.~\eqref{eq:hat_reg}, we have that
\begin{equation*}
    \sum_{t\in[T]}\sum_{f\in[N]^{[C]}}\hat{\rho}_{t,f}\cdot u_L'(a_t,\hat{\pi}(f))-u_L'(a_t,\hat{\beta}_t)\ge\sum_{t\in[T]}\sum_{f\in[N]^{[C]}}\hat{\rho}_{t,f}\cdot u_L'(a_t,\pi(f))-u_L'(a_t,\hat{\beta}_t).
\end{equation*}
It follows that
\begin{align}\label{eq:reg_near_hat_reg_final_ineq}
&\left(\sum_{t\in[T]}\sum_{f\in[N]^{[C]}}\rho_{t,f}\cdot u_L'(a_t,\pi(f))-u_L'(a_t,\beta_t)\right)-\left(\sum_{t\in[T]}\sum_{f\in[N]^{[C]}}\hat{\rho}_{t,f}\cdot u_L'(a_t,\hat{\pi}(f))-u_L'(a_t,\hat{\beta}_t)\right)\nonumber\\
\le&\left(\sum_{t\in[T]}\sum_{f\in[N]^{[C]}}\rho_{t,f}\cdot u_L'(a_t,\pi(f))-u_L'(a_t,\beta_t)\right)-\left(\sum_{t\in[T]}\sum_{f\in[N]^{[C]}}\hat{\rho}_{t,f}\cdot u_L'(a_t,\pi(f))-u_L'(a_t,\hat{\beta}_t)\right)\nonumber\\
=&\sum_{t\in[T]}\sum_{f\in[N]^{[C]}}\rho_{t,f}\cdot (u_L'(a_t,\pi(f))-u_L'(a_t,f))-\sum_{t\in[T]}\sum_{f\in[N]^{[C]}}\hat{\rho}_{t,f}\cdot (u_L'(a_t,\pi(f))- u_L'(a_t,f))\nonumber\\
=&\sum_{t\in[T]}\sum_{f\in[N]^{[C]}}(\rho_{t,f}-\hat{\rho}_{t,f})\cdot (u_L'(a_t,\pi(f))-u_L'(a_t,f))\nonumber\\
\le&\sum_{t\in[T]}\sum_{f\in[N]^{[C]}}|\rho_{t,f}-\hat{\rho}_{t,f}|\cdot |u_L'(a_t,\pi(f))-u_L'(a_t,f)|\nonumber\\
\le&2\cdot \sum_{t\in[T]}\sum_{f\in[N]^{[C]}}|\rho_{t,f}-\hat{\rho}_{t,f}|\le4\cdot\sum_{t\in[T]}\sum_{j\in[N]}\sum_{c\in[C]}|\beta_t(j,c)-\hat{\beta}_t(j,c)|,
\end{align}
where the penultimate inequality is because $u_L'$ has range $[-1,1]$, and the last inequality is by Lemma~\ref{lemma:construct_rho}. Furthermore, let $p_{\min}:=\min_{c\in[C]}p_c$ (note $p_{\min}=\Omega(1)$ because we assume $\G$ has no negligible context in the statement of Theorem~\ref{thm:exploitable-poly-swap-regret}, and Ineq.~\eqref{eq:p_c} will be the only place where we use this assumption), and notice that
\begin{equation}\label{eq:p_c}
    \sum_{t\in[T]}\sum_{j\in[N]}\sum_{c\in[C]} |\beta_t(j,c)-\hat{\beta}_t(j,c)|\le\sum_{t\in[T]}\sum_{j\in[N]}\sum_{c\in[C]} \frac{p_c}{p_{\min}}\cdot |\beta_t(j,c)-\hat{\beta}_t(j,c)|=o(T/p_{\min})=o(T).
\end{equation}
The proof of Claim~\ref{claim:reg_near_hat_reg} finishes by combining Ineq.~\eqref{eq:reg_near_hat_reg_final_ineq} and~\eqref{eq:p_c}.
\end{proof}
\begin{claim}\label{claim:hat_reg'_is_0}
$\widehat{\Reg_{\poly}}'(\G')=0$.
\end{claim}
\begin{proof}[Proof of Claim~\ref{claim:hat_reg'_is_0}]
We will prove
$\sum_{t\in[T]}\eta_t\cdot T\cdot\left(\sum_{f\in[N]^{[C]}}\hat{\beta}_{t,f}\cdot u_L'(a_t,\pi(f))-u_L'(a_t,\hat{\beta}_t)\right)\le 0$ for all $\pi:[N]^{[C]}\to[N]^{[C]}$, which implies $\widehat{\Reg_{\poly}}'(\G')\le 0$ because $\hat{\beta}$ is a feasible solution for the minimization part of the min-max program in Eq.~\eqref{eq:hat_reg'} (and note that $\widehat{\Reg_{\poly}}'(\G')\ge 0$ is obvious because the identity map is a feasible solution for the maximization part).

For any $\pi:[N]^{[C]}\to[N]^{[C]}$, we notice that
\begin{align*}
&\sum_{t\in[T]}\eta_t\cdot T\cdot\left(\sum_{f\in[N]^{[C]}}\hat{\beta}_{t,f}\cdot u_L'(a_t,\pi(f))-u_L'(a_t,\hat{\beta}_t)\right)\\
=&\sum_{t\in[T]\textrm{ s.t. }\eta_t>0}\eta_t\cdot T\cdot\left(\sum_{f\in[N]^{[C]}}\hat{\beta}_{t,f}\cdot u_L'(a_t,\pi(f))-u_L'(a_t,\hat{\beta}_t)\right)\\
=&\sum_{t\in[T]\textrm{ s.t. }\eta_t>0}\eta_t\cdot T\cdot\left(\sum_{f\in[N]^{[C]}}\hat{\beta}_{t,f}\cdot u_L'(a_t,\pi(f))-\sum_{f\in[N]^{[C]}}\hat{\beta}_{t,f}\cdot u_L'(a_t,f)\right)\\
=&T\cdot\sum_{f\in[N]^{[C]}}\left(\sum_{t\in[T]\textrm{ s.t. }\eta_t>0}\eta_t\cdot\hat{\beta}_{t,f}\cdot u_L'(a_t,\pi(f))-\sum_{t\in[T]\textrm{ s.t. }\eta_t>0}\eta_t\cdot\hat{\beta}_{t,f}\cdot u_L'(a_t,f)\right),
\end{align*}
and hence, it suffices to prove that $\sum_{t\in[T]\textrm{ s.t. }\eta_t>0}\eta_t\cdot\hat{\beta}_{t,f}\cdot u_L'(a_t,f)\ge\sum_{t\in[T]\textrm{ s.t. }\eta_t>0}\eta_t\cdot\hat{\beta}_{t,f}\cdot u_L'(a_t,\pi(f))$ for all $f\in[N]^{[C]}$. Recall that by our construction of $\J$, $f$ is a best response to the optimizer's mixed strategy $\alpha$ for any $(\alpha,f)\in\J$, and observe that
\begin{align*}
&\sum_{\substack{t\in[T]\\\textrm{s.t. }\eta_t>0}}\eta_t\cdot\hat{\beta}_{t,f}\cdot u_L'(a_t,f)\\
=&\sum_{\substack{t\in[T]\\\textrm{s.t. }\eta_t>0}}\sum_{(\alpha,\hat{f})\in\J\textrm{ s.t. } \hat{f}=f}\Pr[\alpha,f]\cdot \alpha_{a_t}\cdot u_L'(a_t,f) &&\text{(By definition of $\eta_t,\hat{\beta}_{t,f}$)}\\
=&\sum_{(\alpha,\hat{f})\in\J\textrm{ s.t. } \hat{f}=f}\Pr[\alpha,f]\cdot\sum_{\substack{t\in[T]\\\textrm{s.t. }\eta_t>0}}\alpha_{a_t}\cdot u_L'(a_t,f)\\
=&\sum_{(\alpha,\hat{f})\in\J\textrm{ s.t. } \hat{f}=f}\Pr[\alpha,f]\cdot\sum_{\substack{t\in[T]\\\textrm{s.t. }\eta_t>0}}u_L'(\alpha,f)\\
\ge&\sum_{(\alpha,\hat{f})\in\J\textrm{ s.t. } \hat{f}=f}\Pr[\alpha,f]\cdot\sum_{\substack{t\in[T]\\\textrm{s.t. }\eta_t>0}}u_L'(\alpha,\pi(f)) &&\text{(Since $(\alpha,f)\in\J$)}\\
=&\sum_{\substack{t\in[T]\\\textrm{s.t. }\eta_t>0}}\eta_t\cdot\hat{\beta}_{t,f}\cdot u_L'(a_t,\pi(f)),
\end{align*}
where the last equality follows by the same derivation as the first three equalities.
\end{proof}
To finish the proof of Theorem~\ref{thm:exploitable-poly-swap-regret}, we combine Claim~\ref{claim:reg_near_hat_reg},~\ref{claim:hat_reg_near_hat_reg'},~\ref{claim:hat_reg'_is_0} and get $\Reg_{\poly}(\G')\le\widehat{\Reg_{\poly}}(\G')+o(T)\le\widehat{\Reg_{\poly}}'(\G')+o(T)=o(T)$, which contradicts $\Reg_{\poly}(\G')=\Omega(T)$.
\end{proof}

\subsection{The case of negligible context}\label{section:negligible_p_c}
Now we give an example which consists of a Bayesian game (with a negligible context), a sequence of $T$ actions played by the optimizer, and a sequence of $T$ pure strategies played by a reward-based learner, and we show that the learner has $\Omega(T)$ polytope swap regret but is not exploitable.
\begin{example}\label{example:negligible_context}
\normalfont
In Bayesian game $\G(M,N,C,\D,u_O,u_L)$, we let $M=N=C=2$, and the probabilities $(p_c)_{c\in[C]}$ that specify $\D$ are given as follows: $p_1=o(1)$ (i.e., $p_1$ can be anything as long as it is $o(1)$, such as $0$ or $\frac{1}{T}$) and $p_2=1-p_1$. $u_O$ will not play any role, and for completeness we can let $u_O(i,j,c)=0$ for all $i\in[M],j\in[N],c\in[C]$. $u_L$ is specified as follows: $u_L(i,j,1)=0$ for all $i\in[M],j\in[N]$ and $u_L(1,1,2)=1,\, u_L(1,2,2)=0,\, u_L(2,1,2)=0,\, u_L(2,2,2)=\frac{1}{2}.$

$\G$ is repeated for $T$ rounds. The optimizer plays action $1$ in the first $\frac{T}{2}$ rounds and plays action $2$ in the second $\frac{T}{2}$ rounds. In the first $\frac{T}{2}$ rounds, the learner's reward-based algorithm $\A$ plays action $1$ regardless of the context, and in the second $\frac{T}{2}$ rounds, $\A$ plays action $2$ conditioned on context $1$ and plays action $1$ conditioned on context $2$.
\end{example}

\begin{proposition}\label{proposition:negligible_context}
In Example~\ref{example:negligible_context}, the learner's reward-based algorithm $\A$ has $\Omega(T)$ polytope swap regret, but it is not exploitable.
\end{proposition}
\begin{proof}
Observe that in Example~\ref{example:negligible_context}, the learner uses a pure strategy in the first $\frac{T}{2}$ rounds and uses another pure strategy in the second $\frac{T}{2}$ rounds. We show that the learner would have achieved $\Omega(T)$ more utility by playing a different pure strategy in the second $\frac{T}{2}$ rounds, and hence the learner indeed has $\Omega(T)$ polytope swap regret. Specifically, in Example~\ref{example:negligible_context}, the learner's total utility in the second $\frac{T}{2}$ rounds is $\frac{T}{2}\cdot (p_1\cdot u_L(2,2,1)+p_2\cdot u_L(2,1,2))=0$, but if the learner had played action $2$ regardless of the context, the learner would have achieved utility $\frac{T}{2}\cdot (p_1\cdot u_L(2,2,1)+p_2\cdot u_L(2,2,2))=\frac{(1-o(1))T}{4}$ in the second $\frac{T}{2}$ rounds.

It remains to show that the learner's algorithm $\A$ is not exploitable. Consider any Bayesian game $\G'(M',N,C,\D,u_O',u_L')$ (where $u_O',u_L'$ have range $[-1,1]$) and a sequence of the optimizer's actions $i_1',\dots,i_T'\in[M']$ such that for all $j\in[N]$ and $c\in[C]$, we have $u_L'(i_t',j,c)=u_L(1,j,c)$ for $t\in[\frac{T}{2}]$ and  $u_L'(i_t',j,c)=u_L(2,j,c)$ for $t\in\{\frac{T}{2}+1,\dots,T\}$ (note that this is the only case we need to consider by definition of exploitability -- see Definition~\ref{def:exploitability}). Since $\A$ is reward-based, when faced with the the optimizer's actions $i_1',\dots,i_T'$ in repeated game $\G'$, $\A$ still plays the same pure strategies as in repeated game $\G$, which are specified in Example~\ref{example:negligible_context}. Thus, the total utility of the optimizer in the repeated game $\G'$ is
\begin{align}\label{eq:negaligible_context_example_online_utility}
&\sum_{t\in[\frac{T}{2}]}p_1\cdot u_O'(i_t',1,1)+p_2\cdot u_O'(i_t',1,2)+\sum_{t\in\{\frac{T}{2}+1,\dots,T\}}p_1\cdot u_O'(i_t',2,1)+p_2\cdot u_O'(i_t',1,2)\nonumber\\
=&p_2\cdot\sum_{t\in[T]} u_O'(i_t',1,2)+o(T) \qquad\qquad\qquad\qquad\text{(By $p_1=o(1)$ and $u_O'(i_t',1,1)\le1$)}.
\end{align}
On the other hand, consider the mixed strategy $\alpha\in\Delta([M'])$ of the optimizer corresponding to the following sampling process: sample $t\in[T]$ uniformly at random and play action $i_t'$. We show that playing action $1$ regardless of the context is a best response of the learner to $\alpha$. Notice that we can ignore context $1$ because for all $j\in[N]$, we have $u_L'(i_t',j,1)=0$ for $t\in[T]$ (i.e., against the optimizer's mixed strategy $\alpha$, the learner's utility conditioned on context $1$ can only be zero regardless of which action the learner plays). Thus, we only need to show that playing action $1$ is the learner's best response conditioned on context $2$. To this end, we derive that
\begin{align*}
u_L'(\alpha,1,2)=\frac{1}{T}\cdot\sum_{t\in[T]}u_L'(i_t',1,2)=\frac{1}{T}\cdot\sum_{t\in[\frac{T}{2}]}u_L(1,1,2)+\frac{1}{T}\cdot\sum_{t\in\{\frac{T}{2}+1,\dots,T\}}u_L(2,1,2)=\frac{1}{2}+0=\frac{1}{2},\\
u_L'(\alpha,2,2)=\frac{1}{T}\cdot\sum_{t\in[T]}u_L'(i_t',2,2)=\frac{1}{T}\cdot\sum_{t\in[\frac{T}{2}]}u_L(1,2,2)+\frac{1}{T}\cdot\sum_{t\in\{\frac{T}{2}+1,\dots,T\}}u_L(2,2,2)=0+\frac{1}{4}=\frac{1}{4},
\end{align*}
and it follows that conditioned on context $2$, action $1$ is a better response to $\alpha$ than the only other action $2$. Now observe that the Stackelberg utility $V(\G')$ of $\G'$ is at least the optimizer's expected utility achieved by mixed strategy $\alpha$ against the learner's best response -- playing action $1$ regardless of the context.
That is
\begin{align}\label{eq:negaligible_context_example_stackelberg_utility}
V(\G')\ge p_1\cdot\frac{1}{T}\cdot\sum_{t\in[T]} u_O'(i_t',1,1) + p_2\cdot\frac{1}{T}\cdot\sum_{t\in[T]} u_O'(i_t',1,2)=p_2\cdot\frac{1}{T}\cdot\sum_{t\in[T]} u_O'(i_t',1,2)+o(1),
\end{align}
since $p_1=o(1)$ and $u_O'(i_t',1,1)\in[-1,1]$. It follows by Eq.~\eqref{eq:negaligible_context_example_online_utility} and~\eqref{eq:negaligible_context_example_stackelberg_utility} that the optimizer's total utility in the repeated game $\G'$ is at most $V(\G')\cdot T+o(T)$.
\end{proof}

\ifdefined\unmuteFeweractions
\else
\subsection{Keeping the same number of actions as \texorpdfstring{$\G$}{G} does not preserve exploitability}
In the proof of Theorem~\ref{thm:exploitable-poly-swap-regret} (which we assume the reader is already familiar with), the Bayesian game $\G'(M',N,C,\D,u_O',u_L')$ we constructed has $M'=T$ actions for the optimizer, which means that $\G'$ might have more actions of the optimizer than the original game $\G$ where the learner has high polytope swap regret. Can we prove Theorem~\ref{thm:exploitable-poly-swap-regret} using a different construction that does not blow up the number of the optimizer's actions? In this subsection, we show that keeping the same number of the optimizer's actions as $\G$, which seems to be very natural, does not work. In other words, the answer to the question of~\citet{mansour2022strategizing} would have been negative if we had insisted $M'=M$ in Definition~\ref{def:exploitability}.

As the reader might have noticed, in our construction of $\G'$, even if the actions played by the optimizer in round $t_1$ and round $t_2$ of the repeated game $\G$ were the same (i.e., $i_{t_1}=i_{t_2}$), we would still create two different actions $a_{t_1},a_{t_2}$ that imply the same reward vector from the learner's perspective (i.e., $u_L'(a_{t_1},j, c)=u_L'(a_{t_2},j, c)$ for all $j\in[N],c\in[C]$). Can't we just let $a_{t_1},a_{t_2}$ be the same action in $\G'$ for such cases (and then $M'$ would be reduced to at most $M$)? It turns out the problem is not so simple -- there are cases where unless we allow $a_{t_1},a_{t_2}$ to be different actions (even though $i_{t_1}=i_{t_2}$), regardless of the construction of other components of $\G'$, the optimizer can not exploit the learner in repeated game $\G'$. We elaborate this with the following example.

\begin{example}\label{example:same_action}
\normalfont
In the Bayesian game $\G(M,N,C,\D,u_O,u_L)$, the optimizer has $M=2$ actions, and the learner also has $N=2$ actions. The learner's prior distribution $\D$ is the uniform distribution over $C=2$ contexts, i.e., $p_1=p_2=\frac{1}{2}$. $u_O$ will not play any role, and for completeness we can let $u_O(i,j,c)=0$ for all $i\in[M],j\in[N],c\in[C]$. $u_L$ is specified as follows: $u_L(i,j,1)=0$ for all $i\in[M],j\in[N]$ and $u_L(1,1,2)=2,\,u_L(1,2,2)=1,\,u_L(2,1,2)=0,\,u_L(2,2,2)=2.$

$\G$ is repeated for $T$ rounds. The optimizer plays action $1$ in the first $\frac{2T}{3}$ rounds and plays action $2$ in the last $\frac{T}{3}$ rounds. In the first $\frac{T}{3}$ rounds, the learner's reward-based algorithm $\A$ plays action $2$ conditioned on context $1$ and plays action $1$ conditioned on context $2$, and in the second $\frac{T}{3}$ rounds, $\A$ plays action $1$ conditioned on context $1$ and plays action $2$ conditioned on context $2$, and in the last $\frac{T}{3}$ rounds, $\A$ plays action $2$ regardless of the context.
\end{example}

\begin{table}[ht]
\begin{minipage}{\columnwidth}
\begin{center}
\caption{\label{table:same_action} Illustration of the repeated game $\G$ in Example~\ref{example:same_action}}
\begin{tabular}{ c | c c c } 
 \hline
 optimizer's action &
 $1$ &
 $1$ & 
 $2$ \\
\hline
 learner's pure strategy &
 $\begin{pmatrix}
  \textrm{\color{blue}$\mathbf{0}$} \\ 
  \textrm{\color{blue}$\mathbf{1}$} \\
  \textrm{\color{red}$1$} \\
  \textrm{\color{red}$0$}
\end{pmatrix}$ &
 $\begin{pmatrix}
  \textrm{\color{red}$1$} \\
  \textrm{\color{red}$0$} \\
  \textrm{\color{blue}$\mathbf{0}$} \\
  \textrm{\color{blue}$\mathbf{1}$}
\end{pmatrix}$ &
 $\begin{pmatrix}
  \textrm{\color{blue}$\mathbf{0}$} \\
  \textrm{\color{blue}$\mathbf{1}$} \\
  \textrm{\color{blue}$\mathbf{0}$} \\
  \textrm{\color{blue}$\mathbf{1}$}
\end{pmatrix}$
\end{tabular}
\end{center}
\footnotesize
We decompose the time horizon into three phases each of $\frac{T}{3}$ rounds, which correspond to the right three columns respectively. The first row represents the optimizer's action played in each phase, and the second row represents the learner's pure strategy in each phase -- specifically, in each column vector, the first two coordinates are the probabilities of actions $1$ and $2$ conditioned on context $1$ respectively, and the last two coordinates are the probabilities of actions $1$ and $2$ conditioned on context $2$ respectively. The learner has $\Omega(T)$ polytope swap regret in this repeated game, but if we exchange the first two coordinates of the first column vector with the first two coordinates of the second column vector, the polytope swap regret will be zero.
\end{minipage}
\end{table}

Now we show that the learner has $\Omega(T)$ polytope swap regret in repeated game $\G$ in Example~\ref{example:same_action}, but if we use the same number of optimizer's actions as $\G$ in Bayesian game $\G'$, no matter how we construct other components of $\G'$,  the optimizer can not exploit the learner at all.
\begin{proposition}\label{proposition:same_action}
In Example~\ref{example:same_action}, the learner's reward-based algorithm $\A$ has $\Omega(T)$ polytope swap regret in repeated Bayesian game $\G$.
However, but it is not exploitable in any repeated Bayesian game $\G'$ that has $M'=M$ actions for the optimizer.
\end{proposition}
\begin{proof}[Proof of Proposition~\ref{proposition:same_action}]
Observe that in Example~\ref{example:same_action}, in the second $\frac{T}{3}$ rounds, against the optimizer's action $1$, the learner's algorithm $\A$ plays action $1$ conditioned on context $1$ and plays action $2$ conditioned on context $2$ (and the learner's pure strategy in any other round is different). We show that the learner would have achieved $\Omega(T)$ more utility by playing a different pure strategy in the second $\frac{T}{3}$ rounds, and hence the learner indeed has $\Omega(T)$ polytope swap regret. Specifically, in Example~\ref{example:negligible_context}, the learner's total utility in the second $\frac{T}{3}$ rounds is $\frac{T}{3}\cdot (p_1\cdot u_L(1,1,1)+p_2\cdot u_L(1,2,2))=\frac{T}{6}$, but if the learner had played action $1$ regardless of the context, the learner would have achieved utility $\frac{T}{3}\cdot (p_1\cdot u_L(1,1,1)+p_2\cdot u_L(1,1,2))=\frac{T}{3}$ in the second $\frac{T}{3}$ rounds.

Now consider any Bayesian game $\G'(M,N,C,\D,u_O',u_L')$ (the number of optimizer's actions is $M$ by assumption) and a sequence of the optimizer's actions $i_1',\dots,i_T'\in[M]$. Note that we only need to consider the cases where the resulting reward vectors are exactly the same as in Example~\ref{example:same_action},  by definition of exploitability -- see Definition~\ref{def:exploitability}). Thus, we can w.l.o.g.~assume that $u_L'=u_L$ and $i_t'=1$ for $t\in[\frac{T}{3}]$ and $i_t'=2$ for $t\in\{\frac{T}{3}+1,\dots,T\}$, because any other case (except superficially renaming the optimizer's two actions) would result in some different reward vectors. Since $\A$ is reward-based, when faced with the the optimizer's actions $i_1',\dots,i_T'$ in repeated game $\G'$, $\A$ still plays the same pure strategies as in repeated game $\G$, which are specified in Example~\ref{example:same_action}. Therefore, the expected utility of the optimizer in the repeated game $\G'$ (which we denote by $U_O$) is
\begin{align}\label{eq:same_action_U_O}
U_O&=\sum_{t\in[\frac{T}{3}]}p_1\cdot u_O'(1,2,1)+p_2\cdot u_O'(1,1,2)+\sum_{t\in\{\frac{T}{3}+1,\dots,\frac{2T}{3}\}}p_1\cdot u_O'(1,1,1)+p_2\cdot u_O'(1,2,2)\nonumber\\
&\qquad+\sum_{t\in\{\frac{2T}{3}+1,\dots,T\}}p_1\cdot u_O'(2,2,1)+p_2\cdot u_O'(2,2,2)\nonumber\\
&=T\cdot\bigg(\frac{1}{6}\cdot u_O'(1,2,1)+\frac{1}{6}\cdot u_O'(1,1,2)+\frac{1}{6}\cdot u_O'(1,1,1)+\frac{1}{6}\cdot u_O'(1,2,2)\nonumber\\
&\qquad+\frac{1}{6}\cdot u_O'(2,2,1)+\frac{1}{6}\cdot u_O'(2,2,2)\bigg).
\end{align}

Next, we give two (mixed) strategies $\alpha_1,\alpha_2$ of the optimizer and prove that one of them must have expected utility at least $\frac{U_O}{T}$ when the learner plays the corresponding best response, which implies the proposition by definition of Stackelberg utility. $\alpha_1$ is simply playing action $1$ with probability $1$, and we show that playing action $1$ regardless of the context is a best response of the learner to $\alpha_1$. Indeed, we note that $u_L(i,j,1)=0$ for all $i\in[M],j\in[N]$, and thus, it suffices to show that playing action $1$ is the best response to $\alpha_1$ conditioned on context $2$ as follows
\begin{align*}
u_L(\alpha_1,1,2)=u_L(1,1,2)=2>1=u_L(1,2,2)=u_L(\alpha_1,2,2).
\end{align*}
The expected utility of the optimizer by using $\alpha_1$ against the above-mentioned best response of the learner (which we denote by $U_O^{(1)}$) is
\begin{align}\label{eq:same_action_U_O^{(1)}}
    U_O^{(1)}=\frac{1}{2}\cdot u_O'(\alpha_1,1,1)+\frac{1}{2}\cdot u_O'(\alpha_1,1,2)=\frac{1}{2}\cdot u_O'(1,1,1)+\frac{1}{2}\cdot u_O'(1,1,2).
\end{align}
On the other hand, $\alpha_2$ is playing action $1$ with probability $\frac{1}{2}$ and playing action $2$ with probability $\frac{1}{2}$, and we show that playing action $2$ regardless of the context is a best response of the learner. Similarly to before, we can ignore context $1$ and prove that playing action $2$ is the best response of the learner to $\alpha_2$ conditioned on context $2$ as follows
\begin{align*}
u_L(\alpha_2,2,2)=\frac{u_L(1,2,2)}{2}+\frac{u_L(2,2,2)}{2}=\frac{3}{2}>1=\frac{u_L(1,1,2)}{2}+\frac{u_L(2,1,2)}{2}=u_L(\alpha_1,1,2).
\end{align*}
The expected utility of the optimizer by using $\alpha_2$ against the above-mentioned best response of the learner (which we denote by $U_O^{(2)}$) is
\begin{align}\label{eq:same_action_U_O^{(2)}}
    U_O^{(2)}=\frac{u_O'(\alpha_2,2,1)}{2}+\frac{u_O'(\alpha_2,2,2)}{2}=\frac{u_O'(1,2,1)}{4}+\frac{u_O'(2,2,1)}{4}+\frac{u_O'(1,2,2)}{4}+\frac{u_O'(2,2,2)}{4}.
\end{align}
By Eq.~\eqref{eq:same_action_U_O},~\eqref{eq:same_action_U_O^{(1)}} and~\eqref{eq:same_action_U_O^{(2)}}, we have that $U_O=\left(\frac{1}{3}\cdot U_O^{(1)} + \frac{2}{3}\cdot U_O^{(2)}\right)\cdot T$, and it follows that at least one of $U_O^{(1)}\cdot T$ and $U_O^{(2)}\cdot T$ is no less than $U_O$.
\end{proof}
\fi

\bibliography{cite}

\end{document}